\documentclass[preprint,12pt]{elsarticle}

\usepackage{amsmath}
\usepackage{amsfonts}
\usepackage{amssymb}
\usepackage{amscd}
\usepackage{amsthm}
\usepackage{mathtools}

\usepackage{bbm}
\usepackage{bbding}
\usepackage{mathrsfs}
\usepackage{pifont}
\usepackage{stmaryrd}
\usepackage{txfonts}

\usepackage{array}
\usepackage{booktabs}
\usepackage{diagbox}
\usepackage{threeparttable}

\usepackage{graphicx}
\usepackage{float}
\usepackage{stfloats}
\usepackage{subfig}         
\usepackage{pgfplots}

\usepackage{textcomp}
\usepackage{url}
\usepackage{verbatim}       
\usepackage{color}
\usepackage[left,pagewise]{lineno}
\usepackage[justification=centering]{caption} 

\usepackage{algorithmic}
\usepackage{algorithm}
\usepackage[backref]{hyperref}

\pgfplotsset{compat=1.17}

\newtheorem{theorem}{Theorem}[section]
\newtheorem{lemma}{Lemma}[section]
\newtheorem{definition}{Definition}[section]
\newtheorem{remark}{Remark}[section]

\usepackage{bm}

\journal{}

\begin{document}

\begin{frontmatter}
\title{The efficiency quantum secret sharing schemes on hypercyclic quantum structures}

\begin{abstract}
The generic construction of efficient perfect quantum secret sharing (PQSS) schemes based on hypergraphs with three hyperedges is presented in this paper,and these schemes can reduce the quantum communication cost during share distribution. We first prove that hypercycles with three hyperedges are quantum access structures if and only if they can be classified into 12 non-isomorphic types under hypergraph isomorphism, and these hypercyclic quantum access structures include all four types of hyperstars with three hyperedges. In prior work (Li et al., Inf. Sciences, vol. 664, 2024, 120202), the efficient PQSS schemes were constructed for the above hyperstars using perfect classical secret sharing (PCSS) schemes with optimal information rates(OIR). However, for some hypercyclic access structures, the method described above fails to yield the CPSS schemes with OIR. Then we explicitly construct the CPSS schemes with OIR for these hypercyclic access structures by combining Simmons's geometric method with Shamir's threshold scheme. Finally, We propose the concept of the efficiency of minimally authorized subsets(MAS), and prove that these PQSS schemes achieve the high efficiency.
\end{abstract}

\end{frontmatter}

\section{Introduction}

Quantum secret sharing (QSS) is a quantum cryptographic protocol for securely distributing a secret among multiple participants with quantum information. QSS schemes sharing a classical secret were first proposed by Hillery et al \cite{ref1}, while the schemes sharing a quantum secret were discussed by Cleve et al. [2]. In QSS schemes, the authorized sets of participants are allowed to recover the secret while unauthorized sets ofparticipants are not allowed to have any information about the secret. We call the set composed of all the authorized sets the access structure, denoted by $\varGamma$. An important class of such QSS schemes is the $((t,n))$ quantum threshold schemes \cite{ref1,ref2,ref3,ref4,ref5}. \cite{ref1} first proposed (2,2) and (3,3)-threshold schemes based on GHZ states. Subsequently, Cleve and Wang  introduced (3,4)) and $((t,2t-1))$-threshold access structures \cite{ref2} using  error-correcting codes and local distinguishability of orthogonal multipartite entangled states\cite{ref2,ref5}, respectively. These works further extended to more general $((t,n))$-threshold access structures \cite{ref6, ref7, ref8, ref9, ref10}. However, the above-mentioned $((t,n))$-threshold access structures all correspond to special hypergraph access structure\cite{ref11}.
For example, considering a set of seven participants ${P_1, P_2, ..., P_7}$, let
\begin{equation*}
	\begin{array}{c}
		\varGamma_{0}^{(1)}=\{i j k l \mid 1 \leq i<j<k<l \leq 7\},
	\end{array}
\end{equation*}

\noindent{here,} the authorized subsets $P_i P_j P_k P_l(1=<i<j<k<l\leq 7)$ in $\varGamma^{(1)}_0$ are briefly denoted as $ijkl$.
$\varGamma^{(1)}_0$ belongs to a threshold access structure, and it corresponds to a special hypergraph with the number of vertices being seven, and each hyperedge contains four different vertices. But most quantum network structures emerging from practical scenarios correspond to irregular hypergraphs \cite{ref12,ref13,ref14,ref15}. For example, let $\varGamma_{0}^{(2)}=\{1234,1267,456\}$, then $\varGamma_{0}^{(2)}$ corresponds to a hypergraph which is an irregular one, see Fig.1.
This requires the establishment of an efficient PQSS schemes on the general  quantum access structures.

\begin{figure}[!t]
	\centering
	\includegraphics[width=1.5in]{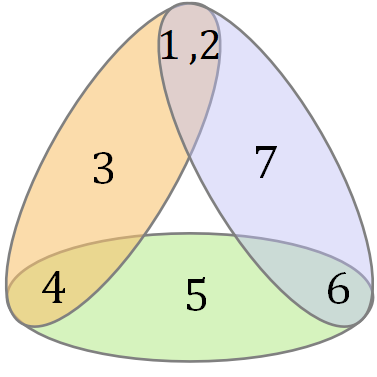}
	\caption{The hypergraph corresponding to minimal access structure $\varGamma^{(2)}_0$.}
	\label{fig_1.1}
\end{figure}

Any hypergraph access structure can be realized using single-particle QSS schemes, however, not all hypergraph access structures can be implemented through multi-particle entangled-state QSS schemes. As the definition of quantum access structures is conditional \cite{ref12}. \cite{ref12} shows that a classical access structure qualifies as a quantum access structure if and only if any two of its authorized sets have non-empty intersection. For instance, as demonstrated in \cite{ref2} that classical $(t,n)$-threshold access structures can serve as quantum threshold access structures if $n\leq 2t$. \cite{ref14} constructed some efficient QSS schemes based on hyperstar quantum access structures with three hyperedges. In this paper, we focus on PQSS schemes which share a classical secret by single particles, and systematically develop the efficient PQSS schemes on hypercycle quantum access structures.

The rest of this paper is organized as follows. In next section, we discuss the related works. we present the PCSS schemes with OIR for the hypercircle quantum access structures in Section 3. Section 4 is devoted to QSS schemes on these quantum access structures. In Section 5, we prove that these schemes are all efficient. And section 6 is the conclusion. We present all 83 hypercycle quantum access structure containing seven participants, and calculate their OIR, and these results are given in the Appendix.

\section{Related works}

Let $F_p$ be a finite field, where $p$ is an odd prime number, $m$ be a positive integer with $p>m$, and $F^{m}_{p}$ be the $m$-dimensional column vector space over $F_p$. We use the notation ${[m]:}=\{1,2,...,m\}$ and ${[i,j]:=\{i,i+1,...,j\}}$. For conciseness, denote the tuple $(A_1, A_2, ...,A_m)$ by $A_{[m]}$.

\subsection{Hypergraph access structures}

The relevant concepts in this subsection are mainly from Ref. \cite{ref11}.

Suppose that there are $n$ participants set denoted as $P=\{P_1, P_2, ..., P_n\}$, $\varGamma$ is the access structure on $P$,  $\varGamma_0$ is a set of all the MAS of $\varGamma$. Obviously, $\varGamma_0 \subseteq \varGamma$. Each MAS of $\varGamma_0$ can be regarded as a hypergraph $H(P, \varGamma_0)$, where $P$ represents the vertex set of $H(P, \varGamma_0)$, and each MAS in $\varGamma_0$ can be regarded as $H(P, \varGamma_0)$.

\cite{ref11} denotes hypergraph as $H(V, E)$, where $V$ is the vertex set and $E$ represents the set consisting of all edges of the hypergraph. In the section 3, we will use $H(P, \varGamma_0)$ instead of $H(V, E)$.

In this paper, all hypergraph refers to connected hypergraph.

\begin{definition}
	\cite{ref11} For a hypergraph $H(V, E)$, if there exists a hyperedge sequence $(E_0, E_1,\cdots, E_{m-1})$ in $E$ such that\\
	(1) $E_{i}\cap E_{(i+1)\mathrm{~mod~}m}\neq \emptyset, i=0, 1, \cdots, m-1$;\\
	(2) For any $i$, $E_{i} \cap E_{j}=\emptyset$, where $j\notin \{(i-1)\mathrm{~mod~}m, i, (i+1){\mathrm{~mod~}}m\}$, $0\leq i\leq m-1$.
	
	Then this hypergraph is called a hypercycle with $m$ hyperedges.
\end{definition}

\begin{definition}
	For a hypergraph $H(V, E)$, if there exists a hyperedge sequence $(E_0, E_1,\cdots, E_{m-1})$ in $E$ such that\\
	(1)$E_{i} \cap E_{i+1} \neq \emptyset, i=0,1, \cdots, m-2$;\\
	(2) For any $i$, $E_{i} \cap E_j=\emptyset$, where $i=0,1,\cdots,m-2$, $j\notin \{i-1, i, i+1\}$.
	
	Then this hyperedge is called a hyperpath with $m$ hyperedges.
\end{definition}

\subsection{Perfect secret sharing scheme}

\begin{definition}
	
	\cite{ref16} A scheme sharing a secret key $K$ for realizing access structure $\varGamma$ in participants set $P$  is called a PCSS if it satisfies the following two conditions:\\
	(Recoverability) for an authorized subset $A\in \varGamma$ of participants, if their shares are pooled together, the value of the secret key $K$ can be determined.\\
	(Secrecy) for an unauthorized subset $B\in \varGamma$ of participants, even if all their shares are collected, no information about the value of $K$ can be determined.
\end{definition}

\begin{definition}
	\cite{ref18} A PQSS for realizing access structure $\varGamma$ on participants $P$ is a quantum cryptographic protocol for securely distributing a secret if it satisfies the following two conditions:\\
	(Recoverability) any authorized set $A\in \varGamma$ can recover the secret i.e. there exists some recovery operation which can decode the secret from the shares in $A$.\\
	(Secrecy) any unauthorized set $B\in \varGamma$ has no information about the secret.
\end{definition}

\begin{definition}
	\cite{ref17}Assume that there is a well-established PCSS scheme $\sum $ that implements a minimal access structure $\varGamma_0$. The information rate of this scheme is defined as
	\begin{eqnarray*}
		\rho(\sum)=\frac{H(K)}{\max _{P_{i} \in P}\left\{H\left(P_{i}\right)\right\}}
	\end{eqnarray*}
	
	\noindent{where $K$ denotes the main key set, and
		$H(K)$ denotes the Shannon's entropy of the probability distribution of the set of shares $(p_s)_{s\in K}$, i.e., $H(K)=-\sum_{s\in K}p_s\mathrm{lb}p_s$.
		
		When the key space and share space are assigned with equal probability, the information rate of this scheme is $\rho=\text{lb}|K|/\max\limits_{P_i\in P}\{\text{lb}|S(P_i)|\}$.
		Here $S(P_i)$ denotes the set where all possible shares of participant $P_i$ are located, $H(P_i)$ denotes the entropy of the probability distribution of the set of shares $S(P_i)$,  where $P_i \in P, i\in [n]$.
	}
\end{definition}

For a given access structure, the upper bound of the information rate in all achievable secret sharing schemes is called the OIR, which was denoted as $\rho^*_C(\varGamma_0)$ \cite{ref19}. It can be shown that the OIR of PCSS scheme that implement an access structure does not exceed 1\cite{ref20}. An access structure with OIR 1 is called an ideal access structure\cite{ref21}.

\begin{lemma} \cite{ref22}
	Assume that $\varGamma$ is an access structure with minimal access structure $\varGamma_0$ .
	$l\geq 1$ is an integer. Let $K$ be a specified key set and $D_h=\{\varGamma_{h,1},\varGamma_{h,2},\cdots,\varGamma_{h,n_h}\}$ be an ideal decomposition of $\varGamma_0$ for the key set $K$, where $1\leq h\leq l$. Let $P_{h,j}$ be the set of participants of the access structure $\varGamma_{h,j}$. For each participant $P_i$, define
	
	\begin{eqnarray*}
		R=\sum_{i=1}^{l}\left|\left\{j: P_{i} \in P_{h, j}\right\}\right|
	\end{eqnarray*}
	
	Then there exists a perfect secret sharing scheme that implements $\varGamma$ with an information rate of $\rho=\frac{l}{R}$, where $R=Max\{R_i:1\leq i \leq n\}$.
\end{lemma}

\begin{lemma}\cite{ref11}
	Let $H(P,\varGamma_0)$ be a hypercycles containing three hyperedges, and let $B_1, B_2, B_3$ be its 2-regions. Then for any $P_i\in B_r, P_j\in B_s$, $H(P_i)+H(P_j) \geq 3H(S)$, where $r,s=1,2,3$ .
\end{lemma}

\subsection{Simmons 's Geometry Construction Method}
The Simmons geometric construction \cite{ref23} for secret sharing scheme can be thought of as a generalization of the Blakley threshold scheme \cite{ref24}. \cite{ref25} improved Simmons' geometric construction \cite{ref23}, resulting in a modified scheme with higher information rate. we describe the scheme as follows \cite{ref25}.

The vectors in $F^{m}_p$ are called points; the cosets of $F^{m}_p$ relative to any 1-dimensional vector subspace, lines; the cosets of $F^{m}_p$ relative to any 2-dimensional vector subspace, planes; the cosets of $F^{m}_p$ relative to any $m-1$ dimensional vector subspace, hyperplanes. More generally, the cosets of $F^{m}_p$ relative to any $r$ -dimensional vector subspace $(0\leq r\leq m)$ are called  $r$-flats. In particular, 0-flats are points, 1-flats are lines, 2-flats are planes, and $(m-1)$-flats are hyperplanes.  An $r$ -flats is said to be incident with an $s$ -flats if the $r$ -flats contains or is contained in the $s$ -flats. Then the point set $F^{m}_p$ with the $r$-flats $(0\leq r\leq m)$ and the incidence relation among them defined above is called the $m$ -dimensional affine space over $F_p$, and denoted by $AG(m, F_p)$.

Let $\alpha_1=(\alpha_{11},\alpha_{12}$,$\cdots$,$\alpha_{1m})$ and $\alpha_2=(\alpha_{21}$,$\alpha_{22}$,$\cdots$,$\alpha_{2m})$ be two distinct points in $AG(m, F_p)$, the line $l$ passing through $\alpha_1$ and $\alpha_2$ is the set of points in the cosets $<\alpha_1-\alpha_2>+\alpha_2$, where $<\alpha_1-\alpha_2>$ is the 1-dimensional vector subspace by the non-zero vector $\alpha_1-\alpha_2$. Hence the line passing through $\alpha_1$ and $\alpha_2$ is the point set
\begin{eqnarray*}
	\left\{\mu \!\left(a_{11}-a_{21},\! a_{12}-a_{22},\! \cdots,\! a_{1m}-a_{2m}\right)\!+\!\left(\!a_{21},\! a_{22},\! \cdots,\! a_{2m}\right)\!\right\},
\end{eqnarray*}
or
\begin{eqnarray*}
	\left\{\mu_{1}\!\left(a_{11},\! a_{12},\! \cdots,\! a_{1m}\right)\!+\!\mu_{2}\left(a_{21},\! a_{22},\! \cdots,\! a_{2 m}\right)\! \mid \mu_{1}\!+\!\mu_{2}\!=\!1\!\right\}.
\end{eqnarray*}
\noindent{where} $\mu\in[0, 1]$.

More generally, the $r$ -flat containing $r+1$ points $\alpha_1, \alpha_2, \cdots, \alpha_{r+1}$, which do not contain in any $(r-1)$-flat , has the parametric representation
\begin{equation}
	\left\{\mu_{1} \alpha_{1}+\mu_{2} \alpha_{2}+\cdots+\mu_{r+1} \alpha_{r+1} \mid \mu_{1}+\mu_{2}+\cdots+\mu_{r+1}=\!1\!\right\}
\end{equation}

For convenience, we'll write this $r$-flat in (1) as \emph{Span} $(\alpha_1, \alpha_2, \cdots, \alpha_{r+1})$.
Suppose that $V_D$ is a fixed line in $AG(m, F_p)$ and Let $V_I$ be a hyperplane such that $|V_I\bigcap V_D|=1$; the key will be the unique point $V_I\bigcap V_D$. Every participant $P_i$ will be given a share consisting of a set $d(P_i)$ of $R_i$ points in $V_I$, and the set $d(P_i)$ chosen in such a way that, for every set $B$, we have
\begin{eqnarray*}
	\operatorname{\emph{Span}}\left\{\sum_{\left\{i: P_{i} \in B\right\}} d\left(P_{i}\right)\right\} \cap V_{D}=\emptyset \Leftrightarrow B \notin \varGamma.
\end{eqnarray*}

In other words, if the participants in the authorization set receive enough information vectors from $V_I$ such that these vectors intersect with those of $V_D$ and just happen to get the key $V_I\bigcap V_D$.

Since both the shares and secret key described in \cite{ref25} are on the publicly known line segment $V_D$, \cite{ref23} simplifies the points in $AG(m, F_p)$ that serve as shares into parameters by extracting the ideas from \cite{ref25}. The specific description is given below. Assume that all $K_{i,j}$ and $K_0$ are public, where $K_{i,j}\in V_I$, and $K_0\in V_I$. Also, let $L_{i,j}$ be a line that contains $K_{i,j}$, be parallel to $V_D$, and $\varepsilon$ be the direction vector of $L_{i,j}$ and $V_D$. Let

\begin{eqnarray*}
	\begin{aligned}
		& V_D = \{K_0 + \lambda \varepsilon \mid \lambda \in F_p\};\\
		& L_{i,j} = \{K_{i,j} + \lambda_{i,j} \varepsilon \mid \lambda \in F_p, i \in [n], j \in [R_i]\}.
	\end{aligned}
\end{eqnarray*}

Let $A$ be a $1\times n$ matrix over $F_p$ that satisfies $A\cdot \varepsilon=1$, and let $\Pi$ be the solution space of $AX=k$, where $k\in F_p$ is the secret key (when different values are taken, the corresponding $\Pi$ is also different). After the straight lines $V_D$ and $L_{i,j}$ intersect the plane $\Pi$, respectively, we can obtain
\begin{equation}
	K = K_0 + \lambda \varepsilon,
\end{equation}
\begin{equation}
	L_{i,j} \cap \Pi = K_{i,j} + \lambda_{i,j} \varepsilon.
\end{equation}

By left-multiplying both sides of (2) and (3) by matrix $A$, respectively, we have
\begin{equation}
	\lambda = k - A K_0,
\end{equation}
\begin{equation}
	\lambda_{i,j} = k - A K_{i,j}.
\end{equation}
for all $i,j$.

Since there is a one-to-one correspondence between $\lambda$ and the point on $V_D$, and a one-to-one correspondence between $\lambda_{i,j}$ and the point on $L_{i,j}$, we can take $\lambda_{i,j}$ as the participant $P_i$'s shares, and take the domain element $\lambda$ as the secret in this scenario. In this way, we only use numbers (instead of $m$-dimensional vectors) as shares and secrets. The distribution rules after such modification are as follows:
\begin{equation}
	\begin{aligned}
		f_{\Pi}(D)=k-AK_{0},\\
	\end{aligned}
\end{equation}
\begin{equation}
	f_{\Pi}(P_{i})=\{k-AK_{i,j}|1\leq j\leq R_{i}\},i\in[n].
\end{equation}

An authorized subset $B$ can calculate this secret key $\lambda$ as a function of $\lambda_{ij}$ as follows:

First, since
\begin{eqnarray*}
	K_0\in span(\sum_{\{i:P_i\in B\}}\bigcup_{j=1}^{R_i}K_{i,j}),
\end{eqnarray*}
\noindent{we} can represent $K_0$ as follows,
\begin{eqnarray*}
	K_0=\sum_{\{i:P_i \in B\}}\sum_{j=1}^{R_i}\mu_{i,j}K_{i,j},
\end{eqnarray*}
\noindent{where} $\sum_{\{i:P_i\in B\}}\sum_{j=1}^{R_i}\mu_{i,j}=1$.

To make the idea of this transformation easier to understand, we will provide an explanation with examples below.

Let $\varGamma^{(3)}_0=\{134, 12, 23\}$.
We set $p=11$, $m=3$. Let $V_I$ be a plane meeting $V_D$ in a point $K_0=(0,0,0)$. And let $\varepsilon=(0,0,1)$, $V_I$ is the plane $z=0$. $K_{1,1}, K_{2,1}, K_{2,2}, K_{3,1}$ and $K_{4,1}$ are the points in $V_I$, where $K_{1,1}=(-1,0,0)$, $K_{2,1}=(0, 1, 0)$, $K_{2,2}=(1, 1, 0)$, $K_{3,1}=(1, -1, 0)$, $K_{4,1}=(1, 1, 0)$. A plane ${\Pi}$ that meets $V_D$
in a point has the equation $ax+ by+z=k$, where $k$ is the secret and $a,b$ are random numbers. Alice distributes the shares to the four participants as follows:

\begin{equation}
	\begin{aligned}
		&P_1\text{~receives~}(k+a);\\
		&P_2\text{~receives~}(k-b,k-a-b);\\
		&P_3\text{~receives~}(k-a+b);\\
		&P_4\text{~receives~}(k-a-b).\\
	\end{aligned}
\end{equation}

In this way, we can obtain
\begin{eqnarray*}
	\begin{aligned}
		& \lambda=k-AK_{0} \\
		& =k-A\sum_{\{i:P_{i}\in B\}}\sum_{j=1}^{R_{i}}\mu_{i,j}K_{i,j} \\
		& =k-\sum_{\{i:P_{i}\in B\}}\sum_{j=1}^{R_{i}}\mu_{i,j}AK_{i,j} \\
		& =k-\sum_{\{i:P_i\in B\}}\sum_{j=1}^{R_i}\mu_{i,j}\left(k-\lambda_{i,j}\right) \\
		& =\!\sum_{\{i:P_i\in B\}}\sum_{j=1}^{R_i}\mu_{i,j}\lambda_{i,j}.
	\end{aligned}
\end{eqnarray*}

Next we use Simmon's geometric construction method and decomposition construction method \cite{ref22} to get the optimal information rate of the access structure $\varGamma^{(2)}_0$.

\textbf{Example 1} For $\varGamma^{(2)}_0=\{1234, 1267, 456\}$, seeing Fig.1, we give an efficient PCSS schemes on  $P=\{P_{[6]}\}$ , here $p=11$. Let
\begin{eqnarray*}
	\begin{aligned}
		& V_{I}=\{(x_{1},x_{2},x_{3},x_{4},x_{5},x_{6})|x_{6}=0\}, \\
		& V_{D}=\{(x_{1},x_{2},x_{3},x_{4},x_{4},x_{5},x_{6})|x_{i}=0,i\in[5]\}, \\
		& \Pi=\{(x_{1},x_{2},x_{3},x_{4},x_{5},x_{6})|A(x_{1},x_{2},x_{3},x_{4},x_{5},x_{6})^{T}=k\},
	\end{aligned}
\end{eqnarray*}

\noindent{where} $A=(a,b,c,d,e,1)$, and $a,b,c,d,e\in F_{11}$ are random number, and ${\Pi}$ intersects with $V_D$ at $k$. Alice selects $9$ points $K_{1,1}, K_{2,1}, K_{2,2}, K_{3,1}, K_{4,1}, K_{5,1}, K_{6,1}, K_{7,1}, K_{7,2}$ in $V_I$ such that any three points among the following 10 points are not collinear,
\begin{eqnarray*}
	K_{1,1}, K_{2,1}, K_{2,2}, K_{3,1}, K_{4,1}, K_{5,1}, K_{6,1}, K_{7,1}, K_{7,2},K_0.
\end{eqnarray*}
\noindent{Here:}
\begin{eqnarray*}
	\begin{aligned}
		&K_{1,1}=(1,2,2,4,5,0)^{T}, K_{2,1}=(2,1,2,2,4,0)^{T}, \\
		&K_{2,2}=(3,0,1,4,3,0)^{T}, K_{3,1}=(4,-1,0,1,1,0)^{T}, \\
		&K_{4,1}=(5,5,-1,0,1,0)^T, K_{5,1}=(5,9,10,5,8,0)^{T}, \\
		&K_{6,1}=(7,4,-3,7,-1,0)^{T}, K_{7,1}=(8,-5,-4,-3,0,0)^{T} \\
		&K_{7,2}=(9,-6,-5,-4,-1,0)^{T}, K_0=(0,0,0,0,0,0)^{T}.
	\end{aligned}
\end{eqnarray*}

Alice publics $K_{1,1}$, $K_{2,1}$, $K_{2,2}$, $K_{3,1}$, $K_{4,1}$, $K_{5,1}$, $K_{6,1}$, $K_{7,1}$,$ K_{7,2}$, $K_0$. Here
\begin{eqnarray*}
	\begin{aligned}
		& V_{D}=\{K_{0}+\lambda\varepsilon|\lambda\in F_{11}\}, \\
		& L_{i,j}=\{K_{i,j}+\lambda_{i,j}\varepsilon|\lambda_{i,j}\in F_{11}\},\mathrm{~where~}i\in[7],j\in[2].
	\end{aligned}
\end{eqnarray*}
\noindent{and} $\varepsilon=(0,0,0,0,0,1)$. Alice calculates
\begin{eqnarray*}
	\lambda_{i,j}=k-AK_{i,j}(i\in [7], j\in [2]),
\end{eqnarray*}

\noindent{and} pass $\lambda_{i,j}$  to $P_i$ by secure channel, where
\begin{eqnarray*}
	\begin{aligned}
		&\lambda_{1,1}=k-AK_{1,1}=k-a-2b-2c-4d-5e, \\
		&\lambda_{2,1}=k-AK_{2,1}=k-2a-b-2c-2d-4e, \\
		&\lambda_{2,2}=k-AK_{2,2}=k-3a-c-4d-3e, \\
		&\lambda_{3,1}=k-AK_{3,1}=k-4a+b-d-e, \\
		&\lambda_{4,1}=k-AK_{4,1}=k-5a-5b+c-e, \\
		&\lambda_{5,1}=k-AK_{5,1}=k-6a+3b-5c+d-5e, \\
		&\lambda_{6,1}=k-AK_{6,1}=k-7a-4b+3c-7d+e, \\
		&\lambda_{7,1}=k-AK_{7,1}=k-8a+5b+4c+3d, \\
		&\lambda_{7,2}=k-AK_{7,2}=k-9a+6b+5c+4d+e.
	\end{aligned}
\end{eqnarray*}

For the MAS $1234$ in $\varGamma^{(2)}_0$, participants $P_1, P_2, P_3, P_4$ solve the following system of equations

\begin{equation*}
\left\{
    \begin{aligned}
        & \mu_{1,1}K_{1,1} + \mu_{2,1}K_{2,1} + \mu_{2,2}K_{2,2} + \mu_{3,1}K_{3,1} + \mu_{4,1}K_{4,1} = K_{0}, \\
        & \mu_{1,1} + \mu_{2,1} + \mu_{2,2} + \mu_{3,1} + \mu_{4,1} = 1,
    \end{aligned}
\right.
\end{equation*}


\noindent{then} obtain $(\mu_{1,1}, \mu_{2,1}, \mu_{2,2}, \mu_{3,1}, \mu_{4,1})=(4,8,7,6,9)$.

For the MAS $1267$ and $456$,  they obtain the following solutions, respectively
\begin{eqnarray*}
	\begin{aligned}
		& (\mu_{1,1},\mu_{2,1},\mu_{2,2},\mu_{6,1},\mu_{7,1},\mu_{7,2})=(4,4,6,1,10,9),\\
		& (\mu_{4,1},\mu_{5,1},\mu_{6,1})=(3,1,8).
	\end{aligned}
\end{eqnarray*}

For the MAS $1234$, participant $P_1, P_2, P_3, P_4$ can obtain the secret
\begin{eqnarray}
	\begin{aligned}
		\lambda &=4  (k-a-2b-2c-4d-5e) \\
		& +8(k-2a-b-2c-2d-4e) \\
		& +7(k-3a-c-4d-3e) \\
		& +6(k-4a+b-d-e) \\
		& +9(k-5a-5b+c-e) \\
		& =k.
	\end{aligned}
\end{eqnarray}

Similarly, the participant in the MAS $1347$ and $2346$ can both obtain the secret key $k$,respectively.
It is easy to verify that the systems of equations in the above form established by any maximal non - authorized subsets in $\varGamma^{(2)}_0$ are all inconsistent. Therefore, they cannot obtain the secret $k$.

\section{The hypercycle quantum access structure with three hyperedges}

This section mainly studies the hypercycle quantum access structures with three hyperedges and their optimal information rate.

Since the minimal access structure $\varGamma_0$ corresponds to a hypergraph $H(P, \varGamma_0)$, for the sake of convenience, the hypergraph $H(P, \varGamma_0)$ corresponding to the minimal access structure $\varGamma_0$ is also denoted as $\varGamma_0$.

\subsection{Types of hypercycle access structures with three hyper-edges}
If $\varGamma_0$ is a connected hypergraph containing three hyperedges $E_1, E_2, E_3$, then there exists $E_i(i\in [3])$ such that $E_{(i-1)\operatorname{~mod~}3}\cap E_i\neq\emptyset$ and $E_{i}\cap E_{(i+1)\mathrm{~mod~}3}\neq\emptyset$.

\begin{theorem}
	Let $\varGamma_0$ be a hypercycle containing $n(n\geq3)$ hyperedges. $\varGamma_0$ is a quantum access structure if and only if $n=3$. And there are 12 types of hypercycles in this case.
\end{theorem}
\begin{proof}
	
	\begin{figure*}[t]
		\centering
		\includegraphics[width=5in]{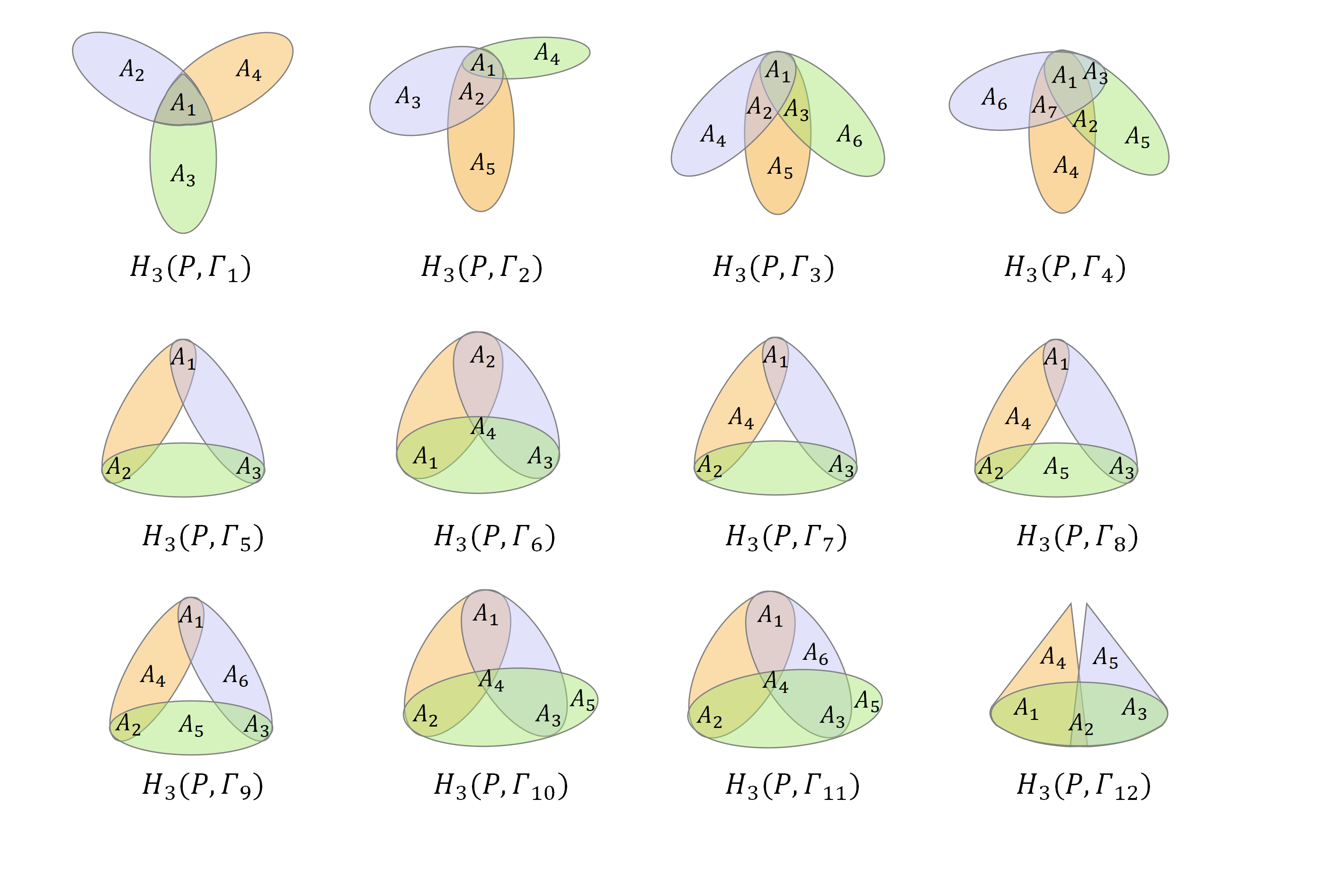}
		\caption{Hypercycle quantum access structure containing 3 hyperedges.}
		\label{fig 2}
	\end{figure*}
	
	\textbf{Necessity: }If $\varGamma_0$ is a quantum access structure, we will prove that $n>3$ is impossible. Assume that $n>3$,  according to the definition of a hypercycle, for $i\in [n]$, then $E_{i}\cap E_{(i+1)\mathrm{~mod~}3}\neq\emptyset$, and for $j\notin\{(i-1) \mathrm{~mod~}3,i,(i+1)\mathrm{~mod~}3\}$, we have $E_{i}\cap E_{j}=\emptyset$. This indicates that the complement of the authorized subset $E_i$ contains the authorized subset $E_j$, which contradicts the fact that $\varGamma_0$ is a quantum access structure.
	
	\textbf{Sufficiency:} If $\varGamma_0$ is a hypercycle containing three hyperedges, then by the definition of a hypercycle, for $i\in[n]$, we have $E_{i}\cap E_{(i+1)\mathrm{~mod~}3}\neq\emptyset$. Therefore, according to the definition of a quantum access structure, $\varGamma_0$ is a quantum access structure.
	
	If $\varGamma_0$ is a hypercycle quantum access structure, then there are a total of 12 types in the sense of isomorphism, and their algebraic representation are as follows:
	\begin{align*}
		\Gamma_{1}  &=\{A_{1}A_{2},A_{1}A_{3},A_{1}A_{4}\}, \\
		\Gamma_{2}  &=\{A_{1}A_{2}A_{4},A_{1}A_{2}A_{5},A_{1}A_{3}\}, \\
		\Gamma_{3}  &=\{A_{1}A_{2}A_{4},A_{1}A_{2}A_{3}A_{5},A_{1}A_{3}A_{6}\}, \\
		\Gamma_{4}  &=\{A_{1}A_{2}A_{4}A_{7},A_{1}A_{2}A_{3}A_{5},A_{1}A_{3}A_{6}A_{7}\}, \\
		\Gamma_{5}  &=\{A_{1}A_{2},A_{1}A_{3},A_{2}A_{3}\}, \\
		\Gamma_{6}  &=\{A_{1}A_{2}A_{4},A_{1}A_{3}A_{4},A_{2}A_{3}A_{4}\}, \\
	\end{align*}
	\begin{align*}
	\Gamma_{7}  &=\{A_{1}A_{2}A_{4},A_{1}A_{3},A_{2}A_{3}\}, \\
	\Gamma_{8}  &=\{A_{1}A_{2}A_{4},A_{1}A_{3},A_{2}A_{3}A_{5}\}, \\
	\Gamma_{9}  &=\{A_{1}A_{2}A_{4},A_{1}A_{3}A_{6},A_{2}A_{3}A_{5}\}, \\
	\Gamma_{10} &=\{A_{1}A_{2}A_{4},A_{1}A_{3}A_{4},A_{2}A_{3}A_{4}A_{5}\}, \\
	\Gamma_{11} &=\{A_{1}A_{2}A_{4},A_{1}A_{3}A_{4}A_{6}, A_{2}A_{3}A_{4}A_{5}\}, \\
	\Gamma_{12} &=\{A_{1}A_{2}A_{3},A_{1}A_{2}A_{4},A_{2}A_{3}A_{5}\}.
	\end{align*}
\end{proof}

\begin{theorem}
	Let $\varGamma_0$ be a hyperstar containing three hyperedges, then it must be a hypercycle.
\end{theorem}

\begin{proof}
	Without loss of generality, we assume that the access structure corresponding to $\varGamma_0$ is $\varGamma_0=\{A_1 A_2, A_1 A_3, A_1 A_4 \}$. Let $E_1=\{A_1A_2\}$, $E_2=\{A_1A_3\}$, $E_3=\{A_1A_4\}$, then there exists a hyperedge sequence $(E_1, E_2, E_3)$ in this hyperstar such that for $i\in [3]$, $E_{i}\cap E_{(i+1)\mathrm{~mod~}3}\neq\emptyset$, and for $j\notin\{(i-1)\mathrm{~mod~}3,i,(i+1)\mathrm{~mod~}3\}$, $E_{i}\cap E_{j}=\emptyset$. According to Definition 2.3, $\varGamma_0$ is a hypercycle.
\end{proof}

The hypergraph representations of the above 12 hypercycle quantum access structure are as follows, as shown in Fig.2, where $H_3(P,\varGamma_i)(i\in[4])$ is a hyperstar..

\subsection{The optimal information rate of hypercycle structures with three hyperedges}
According to Theorems 3.1 and 3.2, the hyperstar quantum access structure containing three hyperedges must be a hypercycle quantum access structure. [13] has given the explicit construction of the PCSS schemes on these hyperstar access structures.
So this subsection constructs PCSS schemes with OPR, and gives these explicit construction for the hypercycle access structures $H_3(P,\varGamma_i)(i\in [5, 12])$ , as shown in Fig.2.

For $P=\bigcup^{m}_{i=1}A_i$,  where $A_i$ and $A_j$ are pairwise dispoint when $i\neq j$, $m\in[4, 7]$. For convenience, let $A_i=\{P^{(i)}_1, P^{(i)}_2, \cdots,P^{(i)}_{m_i}, \}, i\in[m]$.

\begin{theorem}
	If the hypercycle containing three hyperedges is the quantum access structure $\varGamma_i$, then $\rho^*_C(\varGamma_i)=1$, where $i\in[5,6]$.
\end{theorem}

\begin{proof}
	First, an ideal secret sharing scheme is constructed on the access structure $\varGamma_5=\{A_1A_2, A_1A_3, A_2A_3\}$.
	
    Alice selects a secret key $s\in F_p$.

	\textbf{(I)} Shares distribution phase
	
	Alice applies the $(2, 3)$-threshold scheme to distribute a sub-shares to each subset $A_1, A_2, A_3$ respectively. The specific operations are as follows:
	
	Alice selects $x_i$ in $F_p$, uses it as the identities of participants in subsets $A_i,i\in{[3]}$, and public it, where $x_{[3]}$ are different from each other. Alice secretly selects an element $a$ in $F_p$, sets $f(x) = s + ax$, and calculates $y_i=f(x_i), i\in{[3]}$. Further, Alice secretly selects $m_1-1$ elements in $F_p$, denoted as $y^{(1)}_1, y^{(1)}_{2},\cdots, y^{(1)}_{m_1-1}$, and calculates $y^{(1)}_{m_1}=(y_1-\sum^{m_1-1}_{i=1}y^{(1)}_i)\mathrm{~mod~}p$. Then, Alice distributes the value of shares $y^{(1)}_i$ to $P^{(1)}_i, i\in[m_1]$.
	
	Similarly, Alice distributes the value of shares $y^{(l)}_j$ to $P^{(l)}_j$, $j\in m_l$, $l\in\{2,3\}$.
	
	\textbf{(II)} Reconstruct the subkey $y_i$.
	The participants of set $A_i$ apply the $(m_i, m_i)$-threshold scheme share subkey $y_i,i\in[3]$. Accumulate the shares of all participants in $A_1$ to refactor the subkey $y_l$, i.e., $y_{1}=\sum_{i=1}^{m_{1}}y_{i}^{(1)}\mathrm{~mod~}p$, $l\in[2,3]$.
	
	\textbf{(III)} Reconstruct the secret $s$.
	
	First, the participants in the MAS $A_1$ and $A_2$ use $y_1$ and $y_2$ to recover the secret $s$.
	They obtain the following system of linear equations on $F_p$ :
	\begin{equation}
		\begin{cases}s+ax_1=y_1\\s+ax_2=y_2&\end{cases}.
	\end{equation}
	
	Since $x_1$ and $x_2$ are distinct, this system of equations has a unique solution on $F_p$, thus they obtain the secret $s$.
	
	The MAS $A_1A_3, A_2A_3$ can also refactor the secret $s$ in a similar way respectively. It is easy to prove that the unauthorized subset in $\varGamma_5$ cannot obtain any information about the secret $s$.
	
	When the participants in the MAS share the secret $s$, each participant $P^{(i)}_j$ only receives one shares $y^{(i)}_{j_i}, j_i\in[m_i], i\in [3]$, so $\rho^*_C(\varGamma_5)=1$.
	
	Next we will prove that there also exists an ideal secret sharing scheme for $\varGamma_6=\{A_1A_2A_4, A_1A_3A_4, A_2A_3A_4\}$ to implement it. To share the secret $s$ among the participant  set $P$, let $s=s_1+s_2$. Apply the $(|A_4|, |A_4|)$-threshold scheme to the participant in $A_4$ to share $s_1$.
	For the new access structure $\varGamma'=\{A_1A_2, A_1A_3, A_2A_3\}$, share $s_2$ according to the above proof method on $\varGamma_5$. Since the threshold schemes are all ideal, and independent schemes are used among non - intersecting participants, the final scheme for implementing access structure $\varGamma_6$ is also ideal, i.e. $\rho^*_C(\varGamma_6)=1$.
\end{proof}

\begin{lemma}
	Let $\varGamma_i$ be a hypercycle containing three hyperedges, then $\rho^*_C(\varGamma_i)\leq \frac{2}{3}$, $i\in[7 ,11]$.
\end{lemma}

\begin{proof}
	Since $\varGamma_i(i\in[7, 11])$ has at least one hyperedge which has its own independent point set, let $B_r$ and $B_s$ be two non - empty 2-regions of $\varGamma_i$. It is easy to know from Lemma 2.2 that for any $P_i\in B_r$, $P_j\in B_s$, $H(P_i)+H(P_j)\geq 3H(S)$. And then $H(P_i)\geq\frac{3H(S)}{2}$. So $max_{P_{i}\in V}H\left(P_{i}\right)\geq\frac{3H(S)}{2}$, i.e. $\rho^{*}_C(\Gamma_{i})\leq\frac{H(S)}{max_{P_{i}\in V}H(P_{i})}\leq\frac{2}{3},i\in[7,11]$.
\end{proof}

\begin{theorem}
	Let $\varGamma_i$ be a hypercycle containing three hyperedges, then $\rho^{*}(\varGamma_{i})=\frac{2}{3},i\in[7,11]$.
\end{theorem}

\begin{proof}
	According to Lemma 3.1, we have $\rho^*_C(\varGamma_i)\leq\frac{2}{3}(i\in[7, 11])$. Now, we only need to construct an explicit PCSS scheme $\Gamma_{i}(i\in[7, 11])$ with information rate $\frac{2}{3}$ for these hypercycle. First, establish a PCSS scheme for $\varGamma_9$.
	
	Consider the following two decomposition. Let
	\begin{equation}
		\begin{aligned}
			&D_1=\{\varGamma_{1,1}, \varGamma_{1,2}\}, D_2=\{A_1A_2A_4, A_1A_3A_6, A_2A_3A_5\}.
		\end{aligned}
	\end{equation}	
where $\varGamma_{1,1}=\{A_1A_2A_4\}$, $\varGamma_{1,2}=\{A_1A_3A_6, A_2A_3A_5\}$.
	
	Obviously, $D_1$ is an ideal decomposition, where $\varGamma_{1,1}$ can be realized by the $(3, 3)$-threshold scheme, and $\varGamma_{1,2}$ is a hyperstar containing two hyperedges. For $D_2$, we will use Simmons' geometric method to give an implementation scheme as follow.
\end{proof}

\noindent\textbf{Distribution of shares in $\bm{D_1}$ and reconstruction of the secret $\bm{s}$.}

Alice selects 6 different non-zero elements in $F_p$, denoted as $x_{[6]}$, and uses them as the identities of participant in the subset $A_{[6]}$ respectively. Here, $x_{[6]}$ are public. Note that the identity of each participant in $A_i$ is the same, and $s\in F_p$ is the secret that participants of the MAS  will share.

\textbf{(1) For $\bm{\varGamma_{1,1}}$, Alice applies the $\bm{(3,3)}$-threshold scheme to give sub-shares to the participant subset in $\bm{A_1}, \bm{A_2}, \bm{A_4}$ respectively.}

Alice secretly selects two elements $a_1, b_1$ in $F_p$, sets $f_1(x)=s+a_1x+b_1x^2$, and calculates $f_1(x)(i\in\{1,2,4\})$. Further, Alice secretly selects $m_i-1$ elements in $F_p$, denoted as $y_{i,1}^{(1,1)}, y_{i,2}^{(1,1)},\cdots,y_{i,m_i-1}^{(1,1)}$, and calculates $y_{i,m_i}^{(1,1)}=(f_1(x_i)-\sum^{m_i-1}_{j_i=1}y^{(1,1)}_{i,j_i})\mathrm{~mod~} p$. Then, Alice sends the share $y_{i,j_i}^{(1,1)}$ to $P_{j_i}^{(i)}$ through a secure channel, $i\in\{1,2,4\}$, $j_i\in[m_i]$.

Thus each participant in the participant set $P$ has received their own shares.

The participants in $A_i$ apply $(m_i, m_i)$-threshold scheme share subkey $f_1(x_i)$. For example, the shares of all participants in $A_1$ are accumulated to refactor the subkey $f_1(x_1)$, i.e., $f_1(x_1)=\sum^{m_1}_{j_1=1}y^{(1,1)}_{1, j_1} \mathrm{~mod~} p$. Similarly, the shares of all participants in $A_2$ and $A_4$ are accumulated to refactor the subkey $f_1(x_2)$ and $f_1(x_4)$, respectively.

As representatives in the subset $A_1,A_2,A_4$, $P_1^{(1)},P_1^{(2)},P_1^{(4)}$ solve the following system of linear equations:
\begin{equation}
	\begin{cases}
		s+a_1x_1+b_1x_1^2=f_1(x_1)\\
		s+a_1x_2+b_1x_2^2=f_1(x_2)\\
		s+a_1x_4+b_1x_4^2=f_1(x_4)
	\end{cases}.
\end{equation}

Since $x_1,x_2,x_4$ are distinct, this system of equations has a unique solution over $F_p$, thus the secret $s$ is obtained.

\textbf{(2) For \bm{$\varGamma_{1,2}$}, Alice let \bm{$s=s_1+s_2, s_1,s_2\in F_p$}.}

Alice secretly selects $m_3\! -\!1$ elements from $F_p$, denoted as $y_{3,1}^{(1,2)}, y_{3,2}^{(1,2)}, \cdots, y_{3,m_3-1}^{(1,2)}$, and calculates $y_{3,m_3}^{(1,2)}=(s_1-\sum_{j_{3}=1}^{m_{3}-1}y_{3,j_{3}}^{(1,2)})\mathrm{~mod~} p$. Then, she sends the share $y_{3,j_3}^{(1,2)}$ to $P_{j_3}^{(3)}$ through a secure channel, here $j_3\in[m_3]$.

In the $Remove(\varGamma_{1,2},A_{3})=\{A_{1}A_{6},A_{2}A_{5}\}$ share secret $s_2$. Since $A_1A_6$ and $A_2A_5$ are disjoint. Alice sets $f_2(x)=s+a_2x$, secretly selects $m_i-1$ elements from $F_p$, denoted as $y_{i,1}^{(1,2)}y_{i,1}^{(1,2)},\cdots,y_{i,m_{i}-1}^{(1,2)}$, and calculates $y_{i,m_{i}}^{(1,2)}=(f_{2}(x_{i})-\sum_{j_{i}=1}^{m_{i}-1}y_{i,j_{i}}^{(1,2)})\mathrm{~mod~}p,i\in\{1,6\}$. Then, she sends the share $y_{i,j_i}^{(1,2)}$ to $P^{(i)}_{j_3}$ through a secure channel, $j_i\in[m_i], i\in\{1,6\}$.
Alice sets $f_3(x)=s+a_3x$, and secretly selects $m_i-1$ elements from $F_p$, denoted as $y_{i,1}^{(1,2)}y_{i,1}^{(1,2)},\cdots,y_{i,m_{i}-1}^{(1,2)}$, and calculates $y_{i,m_{i}}^{(1,2)}=(f_{3}(x_{i})-\sum_{i,=1}^{m_{i}-1}y_{i,j_{i}}^{(1,2)})\mathrm{~mod~}p,i\in\{2,5\}$. Then, she sends the share $y_{i,j_i}^{(1,2)}$ to $P^{(i)}_{j_3}$ through a secure channel, here $j_i\in[m_i], i\in\{2,5\}$.

$P_1^{(i)}$ acts as the representative in the subset $A_i$, where $i\in\{1,2,3,5,6\}$, and they recover the main secret $s$ in the following way.

The participants in set $A_3$ leverage $(|A_3|,|A_3|)$-threshold scheme share the subsecret $s_1$, that is, they calculate and get $s_1=\sum_{j_{3}=1}^{m_{3}}y_{3,j_{3}}^{(1,2)}\mathrm{~mod~} p$.

$P_1^{(1)}, P_1^{(6)}$ as representatives in subset $A_1, A_6$, respectively, solve the following system of linear equations:

\begin{equation}
	\begin{cases}s_2+a_3x_2=f_2(x_1)\\s_2+a_3x_5=f_2(x_6)&\end{cases}
\end{equation}

Since $x_1, x_6$ are distinct, this system of equations has a unique solution over $F_p$. Thus the subsecret $s_2$ is obtained.

$P_1^{(2)}, P_1^{(5)}$, as representatives in subset $A_2,A_5$,respectively, solve the following system of equations:
\begin{equation}
	\begin{cases}s_2+a_3x_2=f_3(x_2)\\s_2+a_3x_5=f_3(x_5)&\end{cases}
\end{equation}

Since $x_2, x_5$ are distinct, this system of equations has a unique solution over $F_p$, and thus the secret $s_2$ is obtained.

Finally, the MAS $A_1A_3A_6$ (or $A_2A_3A_5$) obtains the secret $s$ by $s=s_1+s_2$.

\noindent{\textbf{ Distribution of share in \bm{$D_2$} and reconstruction of the secret \bm{$k$}.}}

Given $\varGamma_{2,1}=$$\{A_{1}A_{2}A_{4}, A_{1}A_{3}A_{6}, A_{2}A_{3}A_{5}\}$, let $V_{I}=$ $\{(x_{1}$, $x_{2}$, $x_{3}$, $x_{4},x_{5},x_{6})$ $|x_{6}=0\}$,
$V_{D}=\{(x_{1}$, $x_{2},x_{3}$, $x_{4}$, $x_{5}$, $x_{6})|x_{i}$$=0,i=1,2,3,4,5\}$,
$\Pi=\{(x_{1},x_{2}$, $x_{3},x_{4},x_{5}$, $x_{6})|A(x_{1}$,$x_{2},x_{3}$,$x_{4}$,$x_{5}$,$x_{6})^{T}=k\}$,
where $A=(a,b,c,d,e,1)$, $a,b,c,d,e\in F_{11}$.

Suppose $\Pi$ intersects $V_D$ at $k$. Alice selects 8 points $K_{1,1}$, $K_{2,1}$, $K_{3,1}$, $K_{3,2}$, $K_{4,1}$, $K_{4,2}$, $K_{51}$, $K_{6,1}$ on $V_I$ such that any three of the 9 points $K_{1,1}$, $K_{2,1}$, $K_{3,1}$, $K_{3,2}$, $K_{4,1}$, $K_{4,2}$, $K_{51}$, $K_{6,1}$, $K_{7,1}$, $K_{0}$ are non-collinear.

Here
\begin{eqnarray*}
	\begin{aligned}
		&K_{1,1}=(1,a_{1},3,4,5,0),a_{1}\neq2,\\
		&K_{2,1}=(2,1,a_{2},3,a_{4},0),a_{2}\neq2,a_{4}\neq4,\\
		&K_{3,1}=(3,0,b_{2},a_{3},3,0),b_{2}\neq1,a_{3}\neq2,\\
		&K_{3,2}=(4,-1,c_{2},b_{3},2,0),c_{2}\neq0,b_{3}\neq1,\\
		&K_{4,1}=(5,-2,-1,c_{3},b_{4},0),c_{3}\neq0,b_{4}\neq1,\\
		&K_{4,2}=(6,-3,-2,d_{3},c_{4},0),d_{3}\neq-1,c_{4}\neq0,\\
		&K_{5,1}=(7,b_{1},-3,-2,-1,0),b_{1}\neq-4,\\
		&K_{6,1}=(8,-5,-4,-3,d_{4},0),d_{4}\neq-2,\\
		&K_{0}=(0,0,0,0,0,0).
	\end{aligned}
\end{eqnarray*}
Alice makes $K_{1,1},K_{2,1},K_{3,1},K_{3,2},K_{4,1},K_{4,2},K_{5,1},K_{6,1}$ and $K_0$ be public. Since

\begin{eqnarray*}
	\begin{aligned}
		& V_{D}=\{K_{0}+\lambda e|\lambda\in F_{p}\},\\
		& L_{i,j}=\{K_{i,j}+\lambda_{i,j}e|\lambda_{i,j}\in F_{p}\},\quad i\in[6],
	\end{aligned}
\end{eqnarray*}
\noindent{here} $e=(0,0,0,0,0,1)$. Alice calculates $\lambda_{i,j}=k-AK_{i,j},i\in[6],j\in[2],$ and shares $\lambda_{i,j}$ among the participant in the set $A_i$ through the $(m_i,m_i)$-threshold scheme.

Specifically, for $\lambda_{i,j}$, Alice secretly selects $m_i-1$ elements in $F_p$, denoted as $y_{i,1,j}^{(2,1)}y_{i,2,j}^{(2,1)},\cdots,y_{i,m_{i}-1,j}^{(2,1)}$, and calculates $y_{i,m_{i,j}}^{(2,1)}=(\lambda_{i,j}-\sum_{k_{i}=1}^{m_{i}-1}y_{i,k_{i,j}}^{(2,1)})\mathrm{~mod~} p$. Then, she sends the share $y_{i,j,k_{i}}^{(2,1)}$ to the $P_{k_{i}}^{(i)}$ through the secure channel, where $k_{i}\in[m_{i}],i\in[6]$. The participants in the set $A_i$ calculate:
\begin{equation}
	\lambda_{i,j}=\sum_{k_{i}=1}^{m_{i}}y_{i,k_{i},j}^{(2,1)})\mathrm{~mod~} p.
\end{equation}

$P_1^{(1)},P_1^{(2)},P_1^{(4)}$, as representatives in the subsets $A_1,A_2,A_4$, solve the following system of equations
\begin{equation}
	\begin{cases}\mu_{1,1}K_{1,1}+\mu_{2,1}K_{2,1}+\mu_{4,1}K_{4,1}+\mu_{4,2}K_{4,2}=K_{0}\\\mu_{1,1}+\mu_{2,1}+\mu_{4,1}+\mu_{4,2}=1&\end{cases}
\end{equation}

\noindent{then} obtain $(\mu_{1,1},\mu_{2,1},\mu_{4,1},\mu_{4,2})$.

By further substitution, they can obtain $\lambda=\mu_{1,1}\lambda_{1,1}+\mu_{2,1}\lambda_{2,1}+\mu_{4,1}\lambda_{4,1}+\mu_{4,2}\lambda_{4,2}=k$.

Similarly, for the MAS $A_1A_3A_6$ and $A_2A_3A_5$, $k$ can be solved using the above method.

Since this scheme imposes restrictions on the selection of $K_{i,j}$, it ensures that the unauthorized subset in the access structure cannot obtain any information about the secret $k$. Therefore, this scheme is a PCSS scheme.

Furthermore, in this scheme, the shared secret is $(s,k)$, and the share received by the participant $P_1^{(i)}(i\in[6])$ is as follows:
\begin{equation}
	\begin{aligned}
		&P_{1}^{(1)}\text{~receives~}(y_{1,1}^{(1,1)},y_{1,1}^{(1,2)},y_{1,1,1}^{(2,1)});\\
		&P_{1}^{(2)}\text{~receives~}(y_{2,1}^{(1,1)},y_{2,1}^{(1,2)},y_{2,1,1}^{(2,1)});\\
		&P_{1}^{(3)}\text{~receives~}(y_{3,1}^{(1,2)},y_{3,1,1}^{(2,1)},y_{3,1,2}^{(2,1)})\\
		&P_{1}^{(4)}\text{~receives~}(y_{4,1}^{(1,1)},y_{4,1,1}^{(2,1)},y_{4,1,2}^{(2,1)});\\
		&P_1^{(5)}\text{~receives~}(y_{5,1}^{(1,2)},y_{5,1,1}^{(2,1)});\\
		&P_{1}^{(6)}\text{~receives~}(y_{6,1}^{(1,2)},y_{6,1,1}^{(2,1)}).
	\end{aligned}
\end{equation}

Thus, the information rate of this scheme reaches $\frac{2}{3}$. Combining with the above analysis, we know that $\rho^{*}_c\left(H_{3}(P,\Gamma_{9})\right)=\frac{2}{3}$.

By performing the following technical processing on $\varGamma_9$, we can obtain a PCSS scheme on $\varGamma_{i}(i\in\{7,8,10,11\})$ and the OIR as follows.
Setting $A_6=\emptyset$ in $\varGamma_9$, a similar scheme to $\varGamma_9$ can be applied to $\varGamma_8$, we can obtain $\rho^*_C(\varGamma_8)=\frac{2}{3}$. Similarly, by applying a scheme similar to that of $\varGamma_8$ to $\varGamma_7$, $\rho^*_C(\varGamma_7)=\frac{2}{3}$ can be obtained.
Since $Rem(\Gamma_{10},A_{4})=\{A_{1}A_{2}, A_{1}A_{3}, A_{2}A_{3}A_{5}\}\cong\varGamma_{7}$, let $s=(s_{1}+s_{2,1}, s_{1}+s_{2,2})$, where $s_1$ is obtained from the participant in set $A_4$ through $(|A_4|, |A_4|)$-threshold scheme share. Perform two-layer decomposition on $\varGamma_7$ to construct share $(s_{2,1}$, $s_{2,2})$, finally we can get $\rho_C^*(\varGamma_{10})=\frac{2}{3}$. Since $Rem(\Gamma_{11},A_{4})=\{A_{1}A_{2}, A_{1}A_{3}A_{6}, A_{2}A_{3}A_{5}\}\cong\varGamma_{8}$, using a similar method of finding  the OIR of $\varGamma_{10}$, we can get $\rho_C^*(\varGamma_{11})=\frac{2}{3}$.

\begin{theorem}
	If the hypergraph is a hypercycle $\varGamma_{12}$ containing three hyperedges, then $\rho^{*}_C(\varGamma_{12})=\frac{2}{3}$.
\end{theorem}

\begin{proof}
	Since $\varGamma_{12}=\{A_{1}A_{2}A_{3},A_{1}A_{2}A_{4},A_{2}A_{3}A_{5}\}$, and $Rem\{\varGamma_{12},A_{2}\}=\{A_{1}A_{3},A_{1}A_{4}$, $A_{3}A_{5}\}$ is a path, perform a 2-layer decomposition, and constructions on $Rem\{\varGamma_{12},A_{2}\}$ are as follows:
	\begin{eqnarray*}
		D_{1}=\{\varGamma_{1,1},\varGamma_{1,2}\},
	\end{eqnarray*}
	\noindent{where} $\varGamma_{1,1}=\{A_{1}A_{4}\}$, $\varGamma_{1,2}=\{A_{1}A_{3}$, $A_{3}A_{5}\}$.
	
	\begin{eqnarray*}
		D_{2}=\{\varGamma_{2,1},\varGamma_{2,2}\},
	\end{eqnarray*}
	
	\noindent{where} $\varGamma_{2,1}=\{A_{3}A_{5}\}$, $\varGamma_{2,2}=\{A_{1}A_{3}$, $A_{1}A_{4}\}$.
	
	Obviously, both $D_1$ and $D_2$ are two ideal coverages.
	
	Let $s=(s_1+s_{2,1}, s_1+s_{2,2})$. The participant in set $A_2$ uses $(|A_2|,|A_2|)$-threshold scheme share $s_1$, $D_1$ can share $s_{2,1}$, and $D_2$ can share $s_{2,2}$. As a simple calculation, $\rho^{*}(\varGamma_{12})=\frac{2}{3}$ can be obtained.
	
	\textbf{Example 2} Let ${\varGamma_{0}}^{(2)}=\{124,136,235\}$. We set $D_{1}=\{\varGamma_{1,1},\varGamma_{1,2}\}$, $D_{2}=\{\varGamma_{2,1}\}=\{124,136,235\}$.
where $\varGamma_{1,1}=\{124\},\varGamma_{1,2}=\{136,235\}$.
	
	Obviously, $D_1$ is an ideal decomposition. $\varGamma_{1,1}$ can be realized by a $(3,3)$-threshold. And $\varGamma_{1,2}$ is a hyperstar containing two hyperedges. The authorized sets in the two access structures can both recover the key $s$.
	
	A solution for $D_2$ will be presented using Simmons' geometric method. Let
	\begin{eqnarray*}
		\begin{aligned}
			&V_{I}=\{(x_{1},x_{2},x_{3},x_{4},x_{5},x_{6})|x_{6}=0\},\\
			&V_{D}=\{(x_{1},x_{2},x_{3},x_{4},x_{4},x_{5},x_{6})|x_{i}=0,i\in[5]\},\\
			&\Pi=\{(x_{1},x_{2},x_{3},x_{4},x_{5},x_{6})|A(x_{1},x_{2},x_{3},x_{4},x_{5},x_{6})^{T}=k\},
		\end{aligned}
	\end{eqnarray*}
	
	\noindent{where} $A=(a,b,c,d,e,1)$, $a,b,c,d,e\in F_{11}$, and $\Pi$ intersects $V_D$ at $k$. Alice selects eight points $K_1,K_2,K_{3,1},K_{3,2}$, $K_{4,1}$, $K_{4,2}$, $K_{5,1}$, $K_{6,1}$ such that any three of the following ten points $K_{1,1},K_{2,1},K_{3,1},K_{3,2},K_{4,1}$,$K_{4,2},K_{5,1},K_{6,1},K_{0}$
	are not collinear. Here,
	\begin{eqnarray*}
		\begin{aligned}
			&K_{1,1}=(1,2,1,4,0,0), K_{2,1}=(2,1,3,2,4,0),\\
			&K_{3,1}=(3,0,1,1,2,0), K_{3,2}=(4,-1,0,2,3,0),\\
			&K_{4,1}=(5,1,-1,0,1,0), K_{4,2}=(6,1,-2,-1,0,0),\\
			&K_{5,1}=(7,1,-3,3,-1,0), K_{6,1}=(8,4,-4,-3,2,0),\\
			&K_{0}=(0,0,0,0,0).	
		\end{aligned}
	\end{eqnarray*}
	Alice makes $K_{1,1},K_{2,1},K_{3,1},K_{3,2},K_{4,1},K_{4,2},K_{5,1},K_{6,1},K_{0}$ be public.
	Since $V_{D}=\{K_{0}+\lambda e$ $|\lambda\in F_{11}\}$, $L_{i,j}=\{K_{i,j}+$ $\lambda_{i,j}e|\lambda_{i,j}\in F_{11}\},i\in[6].$ Here $e=(0,0,0,0,0,1)$.

Alice calculates $\lambda_{i,j}=k-AK_{i,j}(i\in[6],j\in[2])$ as follows:
	\begin{eqnarray*}
		\begin{aligned}
			&\lambda_{1,1}=k-AK_{1,1}=k-a-2b-1c-4d,\\
			&\lambda_{2,1}=k-AK_{2,1}=k-2a-b-3c-2d-4e,\\
			&\lambda_{3,1}=k-AK_{3,1}=k-3a-c-d-2e,\\
			&\lambda_{3,2}=k-AK_{3,2}=k-4a+b-2d-3e,\\
			&\lambda_{4,1}=k-AK_{4,1}=k-5a-b+c-e,\\
			&\lambda_{4,2}=k-AK_{4,2}=k-6a+b+2c+d,\\
			&\lambda_{5,1}=k-AK_{5,1}=k-7a-b+3c-3d+e,\\
			&\lambda_{6,1}=k-AK_{6,1}=k-8a-4b+4c+3d-2e.
		\end{aligned}
	\end{eqnarray*}
	For the MAS 124 in $\varGamma_{0}^{(2)}=\{124,136,235\}$, participant $P_{1},P_{2},P_{4}$ solves the following system of equations
	\begin{equation}
		\begin{cases}
			\mu_{1,\!1}K_{1,\!1}\!+\!\mu_{2,\!1}K_{2,\!1}\!+\!\mu_{4,\!1}K_{4,\!1}\!+\!\mu_{4,\!2}K_{4,\!2}\!=\!K_{0}\\
			\mu_{1,\!1}+\mu_{2,\!1}+\mu_{4,\!1}+\mu_{4,\!2}=1
		\end{cases}
	\end{equation}
	
	\noindent{and} gets $(\mu_{1,1},\mu_{2,1},\mu_{4,1},\mu_{4,2})=(-1,-6,2,-5)$.
	Similarly, for the MAS 136 and 235, they can get
	\begin{eqnarray*}
		\begin{aligned}
			&(\mu_{1,1},\mu_{3,1},\mu_{3,2},\mu_{6,1})=(2,4,10,7),\\
			&(\mu_{2,1},\mu_{3,1},\mu_{3,2},\mu_{5,1})=(7,-1,-8,3).
		\end{aligned}
	\end{eqnarray*}
	
	For the MAS 124, participant $P_{1},P_{2},P_{4}$ can get the secret $\lambda=-1(k-a-2b-1c-4d)-6(k-2a-b-3c-2d-4e)+2(k-5a-b+c-e)-5(k-6a+b+2c+d)=k$.
	
	For the MAS 136 and 235, the secret $k$ can be solved in the same way.
\end{proof}

According to Theorem 3.4, a PCSS scheme with a two - layer decomposition on hypercycle quantum access structure $\varGamma_0^{(2)}$ is presented. The share of each participant $P_i(i\in[6])$ can be obtained as follows:
\begin{equation}
	\begin{aligned}
		&P_1\text{~receives~}(y_{1,1}^{(1,1)},y_{1,1}^{(1,2)},y_{1,1,1}^{(2,1)});\\
		&P_2\text{~receives~}(y_{2,1}^{(1,1)},y_{2,1}^{(1,2)},y_{2,1,1}^{(2,1)});\\
		&P_3\text{~receives~}(y_{3,1}^{(1,2)},y_{3,1,1}^{(2,1)},y_{3,1,2}^{(2,1)})\\
		&P_4\text{~receives~}(y_{4,1}^{(1,1)},y_{4,1,1}^{(2,1)},y_{4,1,2}^{(2,1)});\\
		&P_5\text{~receives~}(y_{5,1}^{(1,2)},y_{5,1,1}^{(2,1)});\\
		&P_6\text{~receives~}(y_{6,1}^{(1,2)},y_{6,1,1}^{(2,1)}),
	\end{aligned}
\end{equation}
where $y_{i,j}^{(2,1)}=\lambda_{i,j}$. Since this scheme shares the secret $(s,k)$, in the whole scheme,  the OIR of the access structure $\varGamma_0^{(2)}$ is $\frac{2}{3}$, i.e. $\rho_C^*{\varGamma_0^{(2)}}=\frac{2}{3}$.

\section{The QSS scheme based on the hypercycle access structure containing three hyperedges}

In this section, we consider the establishment of the PQSS scheme on the minimal access structure $\varGamma_9$, where $\varGamma_9=\{A_1A_2A_4,A_1A_3A_6,A_2A_3A_5\}$. According to the proof process of Theorem 3.4, access structure $\varGamma_9$ are specifically shown in Fig.2.

\subsection{System model}

The model belongs to a mixture network structure that is made up of classical and quantum structures, where this quantum structure refers to the $p$- dimensional Hilbert space. In this model, $P=\bigcup^6_{i=1}A_i$, where each $A_i$ is derived from the access structure $\varGamma_9$. We assume that each participant from the set $P$ can make quantum states corresponding to the mutually unbiased bases in the Hilbert spaces \cite{ref26}, operate these quantum states by unitary transform \cite{ref13}, and has the capability to measure the state of particles. Furthermore, we assume that the system key $K_j^{(0,i)}$ between Alice and each participant $P_j^{(i)}$, and the initial key $K^{(i,j)}$ between the sets $A_i$ and $A_j$ have been established using the BB84 protocol. We assumed that $K^{(i,j)}$ is owned by each participant from the sets $A_i$ and $A_j$, and it is important to stress that $K_j^{(0,i)}$ and $K^{(i,j)}$ are absolutely secure as system keys.

Since this paper primarily focuses on the relationship between the OIR in PCSS scheme and the efficiency of PQSS scheme corresponding to this PCSS scheme, the model adopted here is a most fundamental one, i.e. we assume that all participants are honest, and Eve is the only external attacker. We adopt the default setting that the recipient has successfully received the sender's information accurately when the measurement results of all quantum states are same, where the sender transmits to each receiver information quantum states with a total of $n_I$ particles and a corresponding decoy-states with $n_E$ particles for detecting Eve's eavesdropping during the share distribution phase,here $n_E \ll p$, and $n_I \ll p$.

The Hash function is used in this paper as a polynomial hash \cite{ref27}. Specifically, we make $PH=\{H_x:F^2_p\rightarrow F_p,|x\in F_p\}$, where $H_x(y)=y_1+y_2x$, $y=(y_1, y_2)\in F_p^2$. It is well known that $PH$ is almost universal in $\frac{2}{p}$.

\subsection{Scheme description}
In this subsection, let $c$ be a primitive element of finite field $F_p$. For convenience, let Alice's public identity be $x_0$, where $x_0\in F_p$.

First, describe the process of using the quantum channel to share the secret $(s,k)$ for the MAS $A_1A_2A_4$ from $\varGamma_9$ in Theorem 3.4. Since $\varGamma_9$ has two layers of decomposition $D_1=\{\varGamma_{1,1}, \varGamma_{1,2}\}$ and $D_2=\{\varGamma_{2,1}\}=\{A_1A_2A_4, A_1A_3A_6, A_2A_3A_5\}$, where $\varGamma_{1,1}=\{A_1A_2A_4\}$, $\varGamma_{1,2}=\{A_1A_3A_6, A_2A_3A_5\}$. Therefore, for the MAS $A_1A_2A_4$, we only need to consider how the participants in three sets $A_1, A_2, A_4$ from $\varGamma_{1,1}$ and $\varGamma_{2,1}$ share the secret $(s,k)$.

\begin{figure*}[!t]
	\centering
	\includegraphics[width=4in]{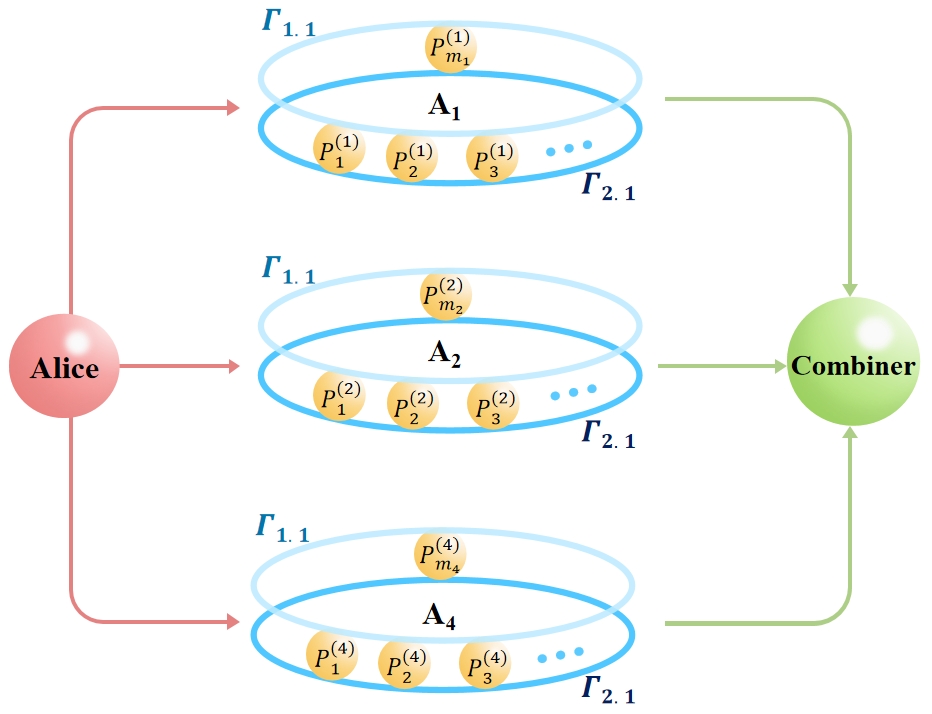}
	\caption{Schematic diagram of the participant in the minimal authorized subset $A_!A_2A_4$ recovering the key $(s,k)$}
	\label{fig_4.1}
\end{figure*}

\textbf{Step 1. Distribution of shares}

Alice manufactures the following particle pairs according to the shares in the proof of Theorem 3.4, as shown in Eq.(20-25). Each participant obtains the shares   through measurement.

\textbf{(1) For the access structure $\varGamma_{1,1}$}

\textbf{(1.1)} Alice first sends the participant $P_{j_i}^{(i)}$ in the set $A_i$ the $n_I$ information particles and the $n_E$ decoy particles,where $j_i \in[m_i],i\in\{1,2,4\}$.
\begin{equation}
	\left\lvert e_{y_{i,j_{i},1}^{(1,1)}}^{H_{K_{j_{i}}^{(0,i)}}(x_{0},x_{i})}\right\rangle,
	\left\lvert e_{y_{i,j_{i},1}^{(1,1)}}^{H_{K_{j_{i}}^{(0,i)}}(x_{0},x_{i})}\right\rangle,\cdots,
	\left\lvert e_{y_{i,j_{i},1}^{(1,1)}}^{H_{K_{j_{i}}^{(0,i)}}(x_{0},x_{i})}\right\rangle,
\end{equation}
\begin{equation}
	\left\lvert e_0^{H_{K_{j_{i}}^{(0,i)}}(x_{0},x_{i})}\right\rangle,
	\left\lvert e_0^{H_{K_{j_{i}}^{(0,i)}}(x_{0},x_{i})}\right\rangle,\cdots,
	\left\lvert e_0^{H_{K_{j_{i}}^{(0,i)}}(x_{0},x_{i})}\right\rangle.
\end{equation}

\textbf{(1.2)} After Alice confirms that $P_{j_i}^{(i)}$ has received the $n_I+n_E$ particles, Alice announces the positions of these $n_E$ decoy particles. $P_{j_i}^{(i)}$ uses the $H_{K_{j_{i}}^{(0,i)}}(x_{0},x_{i})$-th group of unbiased basis to measure these decoy particles (21), and calculates the error rate of the measurement results. If the error rate is higher than the pre-set threshold, $P_{j_i}^{(i)}$ will stop this round and start a new one; otherwise, $P_{j_i}^{(i)}$ uses the $H_{K_{j_{i}}^{(0,i)}}(x_{0},x_{i})$-th group of unbiased basis to measure the information particles received. If the measurement results are all identical to each other, $P_{j_i}^{(i)}$ saves this result as his own share. Otherwise, $P_{j_i}^{(i)}$ requests Alice to start a new round again until the measurement results are identical to each other.

\textbf{(2) For the access structure $\varGamma_{2,1}$}

\textbf{(2.1)} Alice sends the information particles and decoy particles (22-23) to the participant $P_{j_i}^{(i)}$ in the set $A_i$.
\begin{equation}
	\left\lvert e_{y_{i,j_{i},1}^{(2,1)}}^{H_{K_{j_{i}}^{(0,i)}}(x_{0},x_{i})}\right\rangle,
	\left\lvert e_{y_{i,j_{i},1}^{(2,1)}}^{H_{K_{j_{i}}^{(0,i)}}(x_{0},x_{i})}\right\rangle,\cdots,
	\left\lvert e_{y_{i,j_{i},1}^{(2,1)}}^{H_{K_{j_{i}}^{(0,i)}}(x_{0},x_{i})}\right\rangle,
\end{equation}

\begin{equation}
	\left\lvert e_0^{H_{K_{j_{i}}^{(0,i)}}(x_{0},x_{i})}\right\rangle,
	\left\lvert e_0^{H_{K_{j_{i}}^{(0,i)}}(x_{0},x_{i})}\right\rangle,\cdots,
	\left\lvert e_0^{H_{K_{j_{i}}^{(0,i)}}(x_{0},x_{i})}\right\rangle,
\end{equation}
where $j_i \in[m_i],i\in\{1,2,4\}$.

\textbf{(2.2)} Alice sends the information particles and decoy particles (24-25) to the participant $P_{j_i}^{(i)}$ in the set $A_4$.
\begin{equation}
	\left\lvert e_{y_{i,j_{4},2}^{(2,1)}}^{H_{K_{j_{4}}^{(0,4)}}(x_{0},x_{4})}\right\rangle,
	\left\lvert e_{y_{i,j_{4},2}^{(2,1)}}^{H_{K_{j_{4}}^{(0,4)}}(x_{0},x_{4})}\right\rangle, \cdots\!,\!
	\left\lvert e_{y_{i,j_{4},2}^{(2,1)}}^{H_{K_{j_{4}}^{(0,4)}}(x_{0},x_{4})}\right\rangle,
\end{equation}

\begin{equation}
	\left\lvert e_0^{H_{K_{j_4}^{(0,4)}}(x_0,x_4)}\right\rangle,
	\left\lvert e_0^{H_{K_{j_4}^{(0,4)}}(x_0,x_4)}\right\rangle,\cdots,
	\left\lvert e_0^{H_{K_{j_4}^{(0,4)}}(x_0,x_4)}\right\rangle,
\end{equation}

\noindent{where} $j_4\in[m_4]$.

\textbf{(2.3)}$P_{j_i}^{(i)}$ performs the steps similar to the steps (1.2) to ensure that he has obtained the share sent by Alice to himself, where $j_i\in[m_i],i\in\{1,2,4\}$.

\textbf{Step 2. Reconstruction of sub-keys}

For ease of description, we assume that the participants in set $A_i$ pass their respective shares in the protocol in the natural order, where $ i\in\{1,2,4\}$.

\textbf{Case (1):} The reconstruction process of $f_1(x_i)$ for access structure $\varGamma_{1,1}$, where $i\in\{1,2,4\}$.

\textbf{(1.1)} $P_1^{(i)}$ chooses a random number $n_1^{(i)}\in F_p$, then creates the information particles as shown in the following Eq.(26) and creates the decoy particlesas shown in Eq.(27), where $j_i=1$.
\begin{equation}
	\left\lvert e_{log_cK(i,i)}^{n_1^{(i)}}\right\rangle,
	\left\lvert e_{log_cK(i,i)}^{n_1^{(i)}}\right\rangle\!,\cdots\!,\!
	\left\lvert e_{log_cK(i,i)}^{n_1^{(i)}}\right\rangle,
\end{equation}

\begin{equation}
	\left\lvert e_0^{H_{K^{(i,i)}}(x_i,x_i)} \right\rangle,
	\left\lvert e_0^{H_{K^{(i,i)}}(x_i,x_i)} \right\rangle\!,\cdots\!,\!
	\left\lvert e_0^{H_{K^{(i,i)}}(x_i,x_i)} \right\rangle.
\end{equation}
For convenience, denote the state of the information particle in Eq.(26) as $\left\langle \varphi_1^{(i)} \right\lvert$. Then, $P_1^{(i)}$ randomly inserts the decoy particles in Eq.(27) into the information particles in Eq.(26) and sends all particles to $P_2^{(i)}$.

\textbf{(1.2)} After $P_2^{(i)}$ confirms that he has received the $n_I+n_E$ particles, $P_1^{(i)}$ announces the positions of these $n_E$ decoy particles. $P_2^{(i)}$ uses the $H_{K^{(i,i)}}(x_{i},x_{i})$-th group of unbiased basis to measure these decoy particles, and calculates the error rate of the measurement results. If the error rate is higher than the pre-set threshold, $P_2^{(i)}$ will stop this operation and require $P_1^{(i)}$ to start a new round; otherwise, he will continue with the following actions.

\textbf{(1.3)} $P_2^{(i)}$ removes these decoy particles and performs the unitary transformation $U_{y_{i,2}^{(1,1)},0}$ on the $n_I$ information particles. The generated quantum state is denoted as $\left\langle \varphi_{2_i}^{(i)} \right\lvert$. Then, $P_2^{(i)}$ creates $n_E$ decoy particle state $\left\lvert e_0^{H_{K^{(i,i)}}(x_i,x_i)} \right\rangle$ and randomly inserts them into $n_I$ information particles $\left\langle \varphi_2^{(i)} \right\lvert$. Then, $P_2^{(i)}$ sends all particles to $P_3^{(i)}$.

\textbf{(1.4)} The rest can be executed in the same way. Finally, $P_{n_i}^{(i)}$ will obtain $n_I$ information particles, and perform a unitary transformation $U_{y^{(1,1)}_{i,n_i},0}$ on these $n_I$ particles. And denote the obtained each quantum state as $\left\langle \varphi_{n_i}^{(i)} \right\lvert$. Next, $P_{n_i}^{(i)}$ creates $n_E$ decoy particle state $\left\lvert e_0^{H_{K^{(i,i)}}(x_i,x_i)} \right\rangle$ and randomly inserts them into $n_I$ information particles $\left\langle \varphi_{n_i}^{(i)} \right\lvert$. Then, send all particles to $P_1^{(i)}$.

\textbf{(1.5)} $P_{n_i}^{(i)}$ announces the positions of these $n_E$ decoy particles, uses the \\$H_{K^{(i,i)}}(x_{i},x_{i})$-th group of unbiased basis to measure the decoy particle sequence, and calculates the error rate of the measurement results. If the error rate is higher than the pre-set threshold, $P_1^{(i)}$ will stop this operation and require $P_{n_i-1}^{(i)}$ to start a new round; otherwise, he will continue with the following operations.

\textbf{(1.6)} $P_1^{(i)}$ removes these decoy particle and executes unitary transformation $U_{-log_{c}K(i,i)+y_{i,1}^{(1,1)},0}$ on the received information particles. Then $P_1^{(i)}$ uses the $n_I$-th group of unbiased basis to measure the $n_I$ information particles. If the measurement results obtained are identical to each other, retain this measurement result; otherwise, return to step (1.1) to start a new round.

\textbf{Case (2):} The reconstruction of $\lambda_{i,j}$ for the access structure $\varGamma_{2,1}$, $i\in\{1,2,4\}$, $j\in[2]$.

\textbf{(2.1)} $P_1^{(i)}$ chooses a random number $n_1^{(i)}\in F_p$, then manufactures the information particles as shown in the Eq.(24), and manufactures the decoy particles of Eq.(25), where $j_i=1$. For convenience, denote the state of the information particle in Eq.(26) as $\left\langle \varPsi_1^{(i)} \right\lvert$. Then, $P_1^{(i)}$ randomly inserts the decoy particles in Eq.(27) into the information particles Eq.(26), and sends all particles to $P_2^{(i)}$.

\textbf{(2.2)} After $P_2^{(i)}$ confirms receiving the $n_I+n_E$ particles, $P_1^{(i)}$ announces the positions of these $n_E$ decoy particles. $P_2^{(i)}$ uses the $H_{K^{(i,i)}}(x_{i},x_{i})$-th group of unbiased basis to measure the decoy particles received, and calculates the error rate of the measurement results. If error rate is higher than the pre-set threshold, $P_2^{(i)}$ will stop this operation, and require $P_1^{(i)}$ to start a new round; otherwise, he will continue with the following operations.

\textbf{(2.3)} $P_2^{(i)}$ removes these decoy particles and performs the unitary transformation $U_{y_{i,2,j}^{(2,1)},0}$ on the information particles. denote the state of these information particle as $\left\langle \varphi_2^{(i)} \right\lvert$. Then, $P_2^{(i)}$ manufactures $n_E$ decoy particle $\left\langle \varPsi_2^{(i)} \right\lvert$ and randomly inserts them into these $n_I$ information particles $\left\langle \varPsi_2^{(i)} \right\lvert$. After that, $P_2^{(i)}$ sends all particles to $P_3^{(i)}$.

\textbf{(2.4)} The rest can be executed in the same way. Finally, $P_{n_i}^{(i)}$ performs the unitary transformation $U_{y_{i,2,j}^{(2,1)},0}$ on the $n_I$ information particles, and denote the state of these information particle as $\left\langle \varPsi_{n_i}^{(i)} \right\lvert$. Then $P_{n_i}^{(i)}$ creates $n_E$ decoy particles $\left\langle \varPsi_n^{(i)} \right\lvert$ and randomly inserts them into these $n_I$ information particles, and $P_{n_i}^{(i)}$ sends all particles to $P_1^{(i)}$.

\textbf{(2.5)} $P_{n_i}^{(i)}$ announces the positions of these $n_E$ decoy particles, then $P_1^{(i)}$ uses the $H_{K^{(i,i)}}(x_{i},x_{i})$-th group of unbiased basis to measure the decoy particles received, and calculates the error rate of the measurement results. If the error rate is higher than the pre-set threshold, $P_1^{(i)}$ will stop this operation and require $P_{n_i}^{(i)}$ to start a new round; otherwise, he will continue with the following operation.

\textbf{(2.6)} $P_1^{(i)}$ removes these decoy particles and performs the unitary transformation $U_{-log_{c}K(i,i)+y_{i,1,j}^{(2,1)},0}$ on the received information particles. Then, $P_1^{(i)}$ uses the $n_1^{(i)}$-th group of unbiased basis to measure these $n_I$ information particles, if they are identical to each other, then the measurement results are retained; otherwise, it will return to step (2.1) to start a new round.

\begin{remark}
	In Case (2) of this scheme, for $i\in[2]$, $P_1^{(i)}$ performs the unitary transformation $U_{y^{(2,1)}_{i,l,1},0}$; for $i=4$, $P_1^{(4)}$ performs the unitary transformation $U_{y^{(2,1)}_{i,l,j},0}$, where $j\in[2]$.
\end{remark}

\textbf{Step3. Recongstraction of the secret.}

When the combiner needs to access the MAS $A_1A_2A_4$, Alice sets the public identity of the combiner to $y_0$ and the private key between the combiner and $A_i$ to $C^{(0,i)}$, which is owed by each participant from the set $A_i$, where $(y_0, C^{(0,i)})\in F_p^2$.

First, according to the system convention in the sub-key reconstruction of this section, if $P_1^{(i)}(i\in\{1,2,4\})$ uses the $n_1^{(i)}$-th group of unbiased basis to measure the $n_I$ information quantum state obtained, and the measurement results are all equal, then $P_1^{(i)}$ has obtained the sub - key $f_1(x_i)$, or $\lambda_{i,j}$, where $i\in\{1,2,4\}, j\in[2]$.

\textbf{Case (1):} The reconstruction $s$ for the access structure $\varGamma_{1,1}$.

\textbf{(1.1)} According to the operation diagram in Fig.3 of this scheme, for the access structure $\varGamma_{1,1}$, $P_1^{(i)}$ creates the following sequences
\begin{equation}
	\left\lvert e_{f_1(x_i)}^{log_cC^{(0,i)}} \right\rangle,
	\left\lvert e_{f_1(x_i)}^{log_cC^{(0,i)}} \right\rangle, \cdots,
	\left\lvert e_{f_1(x_i)}^{log_cC^{(0,i)}} \right\rangle,
\end{equation}
\begin{equation}
	\left\lvert e_0^{H_{C^{(0,i)}}(y_0,x_i)} \right\rangle,
	\left\lvert e_0^{H_{C^{(0,i)}}(y_0,x_i)} \right\rangle, \cdots,
	\left\lvert e_0^{H_{C^{(0,i)}}(y_0,x_i)} \right\rangle.
\end{equation}
Then, $P_1^{(i)}$ randomly inserts the decoy particles in Eq.(29) into the information particles in Eq.(28), and $P_1^{(i)}$ sends all particles to the combiner.

\textbf{(1.2)} After confirming the receipt of the $n_I+n_E$ particles, $P_1^{(i)} (i\in\{1,2,4\})$ announces the positions of these $n_E$ decoy particles. The combiner uses the $H_{C^{(0,i)}}(y_{0},x_{i})$-th group of unbiased basis to measure the decoy particles received, and calculates the error rate of the measurement results. If the error rate is higher than the pre-set threshold, this operation will be stopped, and $P_1^{(i)}$ will be required to start a new round; otherwise, he will continue with the following operations.

\textbf{(1.3)} The combiner removes these decoy particles, and then uses the $log_cc^{(0,i)}$-th group of unbiased basis to  measure the $n_I$ information particles. If all the obtained measurement results are identical to each other, the measurement result is retained; otherwise, it will return to step (2.1) to start a new round.

\textbf{(1.4)} After obtaining $f_{1}(x_{i})(i\in\{1,2,4\})$, the secret $s$ is obtained by solving the system of equations (12).

\textcolor{blue}{\textbf{Case (2):} The reconstruction $k$ for the access structure $\varGamma_{2,1}$.}

\textbf{(2.1)} According to the operation diagram in Fig.3, $P_1^{(i)}$ creates the following participles
\begin{equation}
	\left\lvert e_{\lambda_{i,j}}^{log_cC^{(0,i)}} \right\rangle,
	\left\lvert e_{\lambda_{i,j}}^{log_cC^{(0,i)}} \right\rangle, \cdots,
	\left\lvert e_{\lambda_{i,j}}^{log_cC^{(0,i)}} \right\rangle,
\end{equation}
\begin{equation}
	\left\lvert e_0^{H_{C^{(0,i)}}(y_0,x_i)} \right\rangle,
	\left\lvert e_0^{H_{C^{(0,i)}}(y_0,x_i)} \right\rangle, \cdots,
	\left\lvert e_0^{H_{C^{(0,i)}}(y_0,x_i)} \right\rangle.
\end{equation}
Then, $P_1^{(i)}, (i\in\{1,2,4\})$ randomly inserts the decoy particles in Eq.(31) into the information particles in Eq.(30), and sends all particles to the combiner.

\textbf{(2.2)} After the combiner confirms the receipt of the  $n_I+n_E$ particles, $P_1^{(i)} (i\in\{1,2,4\})$ announces the positions of these $n_e$ decoy particles. The combiner uses the $H_{C^{(0,i)}}(y_{0},x_{i})$-th group of unbiased basis to measure received decoy particles and calculates the error rate of the measurement results. If error rate is higher than the pre-set threshold, this operation will be stopped, and $P_1^{(i)} (i\in\{1,2,4\})$ will be required to start a new round; otherwise, he will continue with the following operations.

\textbf{(2.3)} The combiner removes these decoy particle. Then, he uses the $log_cc^{(0,i)}$-th group of unbiased basis to measure the $n_I$ information particles. If measurement results are identical to each other, keep the measurement results; otherwise, return to step (2.2) and start a new round.

\textbf{(2.4)} After obtaining $\lambda_{i,j}(i\in\{1,2,4\}, j\in[2])$, the combiner solves the system of Eq.(18) to get the secret $k$.

As can be seen from the above analysis, the MAS $A_1A_2A_4$ has recovered the secret $(s,k)$ in this scheme. For the MAS $A_1A_3A_6$ and $A_2A_3A_5$ in $\varGamma_9$, the combiner can use a similar method as above to reconstruct the secret $(s,k)$.

Using the method of establishing PQSS on $\varGamma_9$, the PQSS on minimal access structure $\varGamma_i(i\in\{7,8,10,11,12\})$ can also be similarly established.

\begin{remark}
We systematically develop these efficient PQSS schemes in this section. The security analysis of this scheme is similar to that of the scheme in \cite{ref14}, so we will not be discussed in detail here.	
\end{remark}

\section{Idealized information rate and efficiency}

\subsection{Idealized information rate}

The information rate $\rho_c(\varGamma_0)$ can be used to measure the efficiency of a PCSS scheme based on the minimal access structure $\varGamma_0$. A higher optimal information rate indicates that there is an optimal scheme for this access structure,which corresponds to lower communication costs during the share distuibution phase\cite{ref28, ref29, ref30}. In this section, the concept of the information rate of the PCSS scheme is extended to that of PQSS scheme, and the definition of the optimal idealized information rate of the minimal access structure is given.

\begin{definition}
	If there is a PQSS scheme $\sum$ which can realize the minimal access structure $\varGamma_0$, and the key space and share space are assigned with equal probability in the $p$ dimensional Hilbert space, then the idealized information rate of this scheme $\sum$ is defined as the ratio
	\begin{eqnarray*}
		\rho_Q(\sum)=\frac{log_p|S|}{\max_{P_i\in P}\{log_p|Q(P_i)|\}}.
	\end{eqnarray*}
\end{definition}
\noindent Here, $log_p|Q(P_i)|$ represents the total number of particles of the space used by share set $S(P_i)$, and $S$ represents the master key set, where $P_i\in P, i\in[n]$.

For a given minimal access structure $\varGamma_0$, the upper bound of the idealized information rate of all PQSS scheme which can realize it is called the optimal idealized information rate, denoted as $\rho^*_Q(\varGamma_0)$.

\begin{remark}
	It is obvious that the idealized information rate essentially ignores the interference of noise in the quantum channel and does not consider the quantum state used for eavesdropping detection.
\end{remark}

 We first calculate the idealized information rate of each participant in this scheme in Section 4. It is easy to see from the distribution phase of this scheme that the number of information particles received by each participant in set $A_i$ is the same. Therefore, we only need to calculate the total number of information particles received by the representative element $P_1^{(i)}$ in the set $A_i$, where $i\in[6]$. For example, in Theorem 3.4, the share distributed by Alice to $P_1^{(1)}$ is $(y_{1,1}^{(1,1)},y_{1,1}^{(1,2)},y_{1,1,1}^{(2,1)})\in F_p^3$. Thus, the total number of information particles received by $P_1^{(1)}$ is $log_p|Q(P_1^{(1)})|=log_pp^{3\Gamma_I}=3\Gamma_I$. In this way, the idealized information rate of $P_1^{(1)}$ is $\rho(P_1^{(1)})=\frac{log_p|S|}{log_p|Q(P_1^{(1)})|}=\frac{2}{3}\frac{1}{\Gamma_I}$, where $log_p|S|=log_pp^2=2$. Using a similar method, the total number of information  particles received by $P_1^{(i)}(i\in[2,6])$ and the idealized information rate can be calculated, seeing Table 1.

\begin{table*}[htbp]
	\centering
	\caption{Total number of information particles and idealized information rate}
	\label{tab:information_particles}
	\setlength{\tabcolsep}{3pt}
	\small
	\begin{tabular}{c *{6}{>{$}c<{$}}}
		\toprule
		& P_1^{(1)} & P_1^{(2)} & P_1^{(3)} & P_1^{(4)} & P_1^{(5)} & P_1^{(6)} \\
		\midrule
		Total number of information particles & 3\Gamma_I & 3\Gamma_I & 3\Gamma_I & 3\Gamma_I & 2\Gamma_I & 2\Gamma_I \\
		Idealized information rate  & 2 \frac{1}{\Gamma_I} & \frac{2}{3}\frac{1}{\Gamma_I} & \frac{2}{3}\frac{1}{\Gamma_I} & \frac{2}{3}\frac{1}{\Gamma_I} & \frac{1}{\Gamma_I} & \frac{1}{\Gamma_I} \\ 
		\bottomrule
	\end{tabular}
\end{table*}

According to the definition 5.1, the idealized information rate of this scheme is $\frac{2}{3}\frac{1}{\Gamma_I}$. Furthermore, since this scheme is a PQSS scheme based on the access structure $\varGamma_9$ with OIR $\frac{2}{3}$, so we have
\begin{eqnarray*}
	\rho_Q^*(\Gamma_9)=\frac{2}{3}\frac{1}{\Gamma_I}.
\end{eqnarray*}

Using the same method, we can obtain that
\begin{eqnarray*}
	\rho_Q^*(\Gamma_i)=\frac{2}{3}\frac{1}{\Gamma_I}.
\end{eqnarray*}
where $i\in\{7,8,10,11,12\}$.

Since the PCSS scheme built on $\varGamma_i$ with $\rho_C^*(\Gamma_i)$ requires less space for share proportions during the secret distribution phase compared to other PCSS schemes, and the PQSS scheme proposed in this paper is built on $\varGamma_i$ with $\rho_C^*(\Gamma_i)$, it is evident that the PQSS scheme with $\rho_Q^*(\Gamma_i)$ achieves the lowest quantum communication cost during the share distribution phase among all schemes, where $i\in [7,12]$.

\subsection{Efficiency}

The efficiency of quantum secret sharing schemes has been discussed in multiple papers \cite{ref31, ref32, ref33, ref34, ref35}. \cite{ref31, ref32, ref33} Mainly discussed the efficiency during the secret reconstruction phase of the scheme, which is defined by the maximum communication cost of each authorized subset in the access structure of the scheme. Therefore, the efficiency of the entire process of the scheme has not been studied. \cite{ref35} gives the definition of efficiency in the quantum key distribution scheme. This definition follows the approximation of the Holevo bound, requiring the optimization of quantum resources and the minimization of classical communication overhead. If the scheme in \cite{ref35} does not involve classical communication, the efficiency of the quantum key distribution scheme can be simplified to the following form:
\begin{eqnarray*}
	\eta=\frac{c}{q},
\end{eqnarray*}

\noindent{where} $c$ represents the number of shared classical bits, and $q$ represents the total number of particles prepared in the scheme, including the total information number of particles $q_I$ and decoy particle $q_E$.

In essence, from the perspective of the MAS contained in the access structure, the scheme in \cite{ref35} can be regarded as the case that its access structure contains only one MAS, i.e. the efficiency of its scheme is precisely that of its unique MAS. Based on this, this paper extends the concept of efficiency from \cite{ref35}, introducing the efficiency of any MAS within a general access structure.

\begin{definition}
	Let $\varGamma_0=\{MAS_1,MAS_2,\cdots,MAS_l\}$ be a quantum access structure, where $MAS_i(i\in[l])$ is the minimal authorized subset, and $\varSigma$ be a PQSS scheme established on $\varGamma_0$. Then the efficiency of the $MAS_i$  in scheme $\varSigma$ is defined as
	\begin{eqnarray*}
		\eta(MAS_i)=\frac{c^{(i)}}{q_{I}^{(i)}+q_{E}^{(i)}+b^{(i)}},
	\end{eqnarray*}
	\noindent{where} $c^{(i)}$ represents the number of shared classical bits of the $MAS_i$, $q_{I}^{(i)}$ represents the total number of information particles required in the $MAS_i$ $q_{E}^{(i)}$ dose the total number of particles required for detecting Eve's eavesdropping, and $b^{(i)}$ dose the total number of classical communication bits in this scheme.
\end{definition}

In this scheme proposed in section 4, since each participant can determines whether the information received is accurate, i.e. when the measurement results of $\varGamma_I$ information particle are identical to each other, he can be certain that the state of his receiving information particle is correct; in the same way, the receiver can detect whether there is an external eavesdropper Eve interfering during the process of his receiving decoy particles. Therefore, there is no need for classical communication in this scheme. So we can get that $b^{(i)}=0$.

We presents the efficiency of three minimal authorized subsets as follows.
\begin{equation}
	\begin{aligned}
		&\eta(A_1A_2A_4)\\
		&=\frac{c_{A_1A_2A_4}}{q_{A_1A_2A_4}}\\
		&=\frac{2}{3}.\frac{1}{(\Gamma_{I}+\Gamma_{E})(n_{1}+n_{2}+n_{4}+3)}\\
		&=\frac{2}{3}\frac{1}{\Gamma_{I}}\frac{1}{(1+d)(n_{1}+n_{2}+n_{4}+3)}\\
		&=\rho_Q^*(\Gamma_9)\frac{1}{(1+d)(n_1+n_2+n_4+3)}
	\end{aligned}
\end{equation}
Where  $q_{E}^{(i)}=n_13\Gamma_I+3n_2\Gamma_I+3n_4\Gamma_I+3\Gamma_I$, $q_{E}^{(i)}=3n_1\Gamma_E+3n_2\Gamma_E+3n_4\Gamma_E+3\Gamma_E$   .Here, $d=\Gamma_{E}/\Gamma_{I}$.
Using the above calculation method, we can obtain:
\begin{equation}
	\eta(A_1A_3A_6)==\rho_Q^*(\Gamma_9)\frac{1}{(1+d)(n_1+n_3+n_6)},
\end{equation}
\begin{equation}
	\eta(A_2A_3A_5)==\rho_Q^*(\Gamma_9)\frac{1}{(1+d)(n_2+n_3+n_5)}.
\end{equation}

From Eq.(32-34), it can be seen that if the idealized information rate of $\varGamma_9$ is higher, then the efficiency of each of its MAS is also higher. This also indicates that the idealized information rate can directly reflect the efficiency of different MAS in sharing the secret.

From the above analysis, it can be known that this paper uses the parameters in the PCSS scheme and single photons to establish a PQSS scheme on the minimal access structure $\varGamma_9$. From Eq.(31-33), it can be known that the efficiency of the MAS in minimal access structure $\varGamma_9$ has all reached the optimal values.

Using the same method, the optimal idealized information rate on minimal access structure $\varGamma_i(i\in\{7,8,10,11,12\})$ and the optimal efficiency of their MAS can be obtained.

\section{Conclusion}
 
A generic construction method for the efficient PQSS scheme on hypercycle quantum access structure is presented in this paper, and we uses the parameters in the classical schemes established on these access structures to generate the corresponding single photons. These schemes can provide a practical theoretical guarantee for the communication of light quantum network structure. Because the entire process of these schemes does not require any quantum storage, and the data measured by participant can be saved, they are relatively easier to implement compared to entangled states schemes under the existing experimental conditions,

In essence, these quantum access structures can also be used to establish secret sharing schemes using multi-particle entangled states. Some scholars have also studied a large number of high-dimensional multi-particle entangled states and used the asymmetry exhibited by these high-dimensional multi-particle entangled states to improve the quantum communication efficiency of hierarchical key structures \cite{ref15,ref16}. However, These studies primarily aim to investigate how to prepare entangled states and utilize them to design efficient QSS schemes ; some scholars have studied the efficiency of QSS schemes during the reconstruction phase, but they have not discussed the efficiency of the schemes throughout the entire phase \cite{ref31, ref32, ref33}; \cite{ref34} studied the efficiency of QSS on a class of restricted access structures, not general access structures. However, at present, no systematic method has been given to establish an efficient PQSS scheme on general quantum access structure using multi-particle entangled states.

\section{Appendix}

This section presents the hyper-circle access structure with seven participants and three hyper-edges and their optimal information rates, as shown in Table 2. The access structures corresponding to each serial number are distributed as follows:
$\varGamma^1(1-7), \varGamma^2(8-18), \varGamma^3(19-22), \varGamma^4(23), \varGamma^5(24-27), \varGamma^6(28-34), \varGamma^7(35-47), \varGamma^8(48-56), \varGamma^{9}(57-58), \varGamma^{10}(59-69), \varGamma^{11}(70-74), \varGamma^{12}(75-83)$, where $\varGamma^i$ denotes the i-th isomorphic class of the hyper-circle access structure with three hyper-edges, $i \in {[12]}$.

\begin{table*}[t!]
	\centering
	\caption{Hyper-circle access structures with 7 participants and 3 hyper-edges}
	\label{tab:6.1}
	\scalebox{0.85}{
		\begin{tabular}{c l | c l | c l}
			\toprule
			No. & Access Structure $\varGamma$ & No. & Access Structure $\varGamma$ & No. & Access Structure $\varGamma$ \\
			\midrule
			1 & \{12345,16,17\} & 29 & \{123456,123457,567\} & 57 & \{1234,1267,456\}  \\
			2 & \{1234,156,17\} & 30 & \{123456,123457,234567\} & 58 & \{1234,167,456\} \\
			3 & \{123,145,167\} & 31 & \{123456,123457,4567\} & 59 & \{12345,23456,1267\} \\
			4 & \{12345,126,127\} & 32 & \{123456,12347,4567\} & 60 & \{12345,12367,3467\} \\
			5 & \{1234,1256,127\} & 33 & \{123456,12347,34567\} & 61 & \{12345,12367,23467\} \\
			6 & \{12345,1236,1237\} & 34 & \{123456,234567,12347\} & 62 & \{12345,1267,2367\} \\
			7 & \{12345,12346,12347\} & 35 & \{1234,4567,3567\} & 63 & \{1123456,12347,457\} \\
			8 & \{12345,12346,17\} & 36 & \{12345,23457,167\} & 64 & \{123456,12347,23457\} \\
			9 & \{126,12345,17\} & 37 & \{12345,1267,3457\} & 65 & \{123456,127,237\} \\
			10 & \{1234,156,157\} & 38 & \{12345,34567,127\} & 66 & \{123456,12347,3457\} \\
			11 & \{12345,1236,17\} & 39 & \{12345,167,567\} & 67 & \{1123456,1237,3457\} \\
			12 & \{1234,1235,167\} & 40 & \{12345,1267,567\} & 68 & \{123456,1237,347\} \\
			13 & \{1234,1256,17\} & 41 & \{12345,1267,3467\} & 69 & \{123456,1237,2347\} \\
			14 & \{1234,1235,1267\} & 42 & \{123456,12347,67\} & 70 & \{12345,1237,3467\} \\
			15 & \{1234,125,167\} & 43 & \{123456,17,67\} & 71 & \{12345,23467,1237\} \\
			16 & \{12345,12346,127\} & 44 & \{123456,3457,127\} & 72 & \{12345,23467,127\} \\
			17 & \{12345,1236,127\} & 45 & \{123456,127,67\} & 73 & \{12345,2367,127\} \\
			18 & \{1234,1235,1267\} & 46 & \{123456,3457,67\} & 74 & \{12345,1456,1267\} \\
			19 & \{12345,1236,147\} & 47 & \{123456,127,567\} & 75 & \{1234,1235,3467\} \\
			20 & \{1234,1256,137\} & 48 & \{1234,4567,127\} & 76 & \{1234,2567,125\} \\
			21 & \{12345,126,137\} & 49 & \{1234,1267,4567\} & 77 & \{12345,12346,457\} \\
			22 & \{12345,1236,1247\} & 50 & \{1234,4567,17\} & 78 & \{12345,12346,23457\} \\
			23 & \{1235,1367,1247\} & 51 & \{12345,1237,567\} & 79 & \{12345,126,267\} \\
			24 & \{12345,34567,1267\} & 52 & \{12345,567,17\} & 80 & \{12345,12346,3457\} \\
			25 & \{123456,123457,67\} & 53 & \{12345,34567,17\} & 81 & \{12345,1236,3457\} \\
			26 & \{123456,34567,127\} & 54 & \{12345,127,567\} & 82 & \{12345,1236,367\} \\
			27 & \{123456,1237,4567\} & 55 & \{12345,4567,127\} & 83 & \{12345,1236,2367\} \\
			28 & \{12345,12367,34567\} & 56 & \{12345,4567,17\} &  &  \\
			\bottomrule
		\end{tabular}
	}
\end{table*}

In Table 2, the access structures numbered 1 - 23 belong to the hyper-stars, and the specific implementation schemes with OIR have been given in \cite{ref13}. According to Theorem 3.3, the OIR of the hyper-circles numbered 24 - 34 is 1; according to Theorem 3.4, the OIR of the hyper-circles numbered 35 - 74 is $\frac{2}{3}$; according to Theorem 3.5, the OIR of the hyper-circles numbered 75 - 83 is $\frac{2}{3}$.
	
	Using the methods in Section 4, the efficient PQSS schemes for the above small-scale quantum access structures can be established. These schemes also provide a powerful tool for constructing PQSS schemes on quantum access structures formed by a larger number of participants.
%

\section*{Acknowledgment}
This work was supported by the National Natural Science Foundation of China 12201484.

\end{document}